%% file: hitting-proof.tex
\documentclass[11pt]{article}
\usepackage{fullpage,amsmath,amssymb,amsthm,url,graphicx,microtype,xcolor,tikz}
\usepackage[colorlinks=true,citecolor=blue,linkcolor=blue,urlcolor=blue]{hyperref}
\usepackage{authblk}
\newcommand{\scterm}[1]{{\textrm{\sc #1}}}
\newcommand{\hit}{\scterm{Hitting}}
\newcommand{\tres}{\scterm{tl-Res}}
\newcommand{\res}{\scterm{Res}}
\newcommand{\resp}{{\res}$(\oplus)$}
\newcommand{\tresp}{{\tres}$(\oplus)$}
\newcommand{\rlin}{{\res}\scterm{(Lin)}}
\newcommand{\F}{\scterm{Frege}}
\newcommand{\EF}{\scterm{Extended {\F}}}
\newcommand{\extres}{\scterm{Extended {\res}}}
\newcommand{\extpc}{\scterm{Ext-PC}}
\newcommand{\fpc}{\scterm{$\mathcal{F}$-PC}}
\newcommand{\cp}{\scterm{CP}}
\newcommand{\pc}{\scterm{PC}}
\newcommand{\pcr}{\scterm{PCR}}
\newcommand{\ns}{\scterm{NS}}
\newcommand{\nsr}{\scterm{NSR}}
\newcommand{\ls}{\scterm{LS}}
\newcommand{\sa}{\scterm{SA}}
\newcommand{\sos}{\scterm{SoS}}
\newcommand{\scs}{\scterm{SCS}}
\newcommand{\hitres}{{\hit} {\res}}
\newcommand{\hitp}{{\hit}$(\oplus)$}
\newcommand{\hitodd}{\scterm{Odd} {\hit}}
\newcommand{\hitk}[1]{\hit$[#1]$}
\newcommand{\rect}{\mathsf{rect}}
\newcommand{\rectone}[1]{\rect(#1,#1)}
\newcommand{\recttwo}[2]{\rect(#1,#2)}
\newcommand{\prt}{\mathsf{prt}}
\newcommand{\supp}{\mathrm{supp}}
\newcommand{\disj}{\scterm{Disj}}
\newcommand{\PM}{\scterm{PM}}
\newcommand{\forxorPM}{\textsc{PM}^\oplus}
\newcommand{\xorPM}{\textsc{Search-PM}^\oplus}
\newcommand{\Alice}{\mathcal{A}}
\newcommand{\Bob}{\mathcal{B}}
\newcommand{\AND}{{\scterm{AND}}}
\newcommand{\OR}{{\scterm{OR}}}
\newcommand{\indexing}{\scterm{Indexing}}
\newcommand{\fcs}{f_{\mathrm{CS}}}
\newcommand{\dagw}{\mathsf{w}}
\newcommand{\leafw}{\dagw_{\mathrm{out}}}
\newcommand{\cert}{\mathsf{C}}
\newcommand{\ucert}{\mathsf{UC}}
\newcommand{\D}{\mathsf{D}}
\newcommand{\R}{\mathsf{R}}
\newcommand{\uw}{\mathsf{uw}}
\DeclareMathOperator{\vars}{vars}

\newcommand{\olnot}[1]{\overline{#1}}
\newcommand{\set}[1]{\{#1\}}
\newcommand{\union}{\cup}

\newcommand{\restrict}[2]{{#1}{\upharpoonright_{#2}}}
\newcommand{\GF}{\mathrm{GF}}
\newcommand{\Search}{\mathrm{Search}}
\newcommand{\E}{\mathbf{E}}
\newcommand{\TODOJ}[1]{}%\textcolor{red}{[Last changes: #1]}}

\newcommand{\exarrow}[1]{%
\raisebox{-.5ex}{\scalebox{0.8}{%
\protect\tikz[every node/.append style={rectangle, draw=black, >=stealth }, shorten >=1mm, shorten <=1mm]{
\protect\draw (0,0) node (A) {\small A};
\protect\draw (1.5,0) node (B) {\small B};
\protect\draw[#1, line width=0.5mm] (A) to (B);
}}%
}} 

\numberwithin{equation}{section}
\newtheorem{theorem}{Theorem}[section]
\newtheorem{lemma}[theorem]{Lemma}
\newtheorem{proposition}[theorem]{Proposition}
\newtheorem{corollary}[theorem]{Corollary}
\newtheorem{remark}[theorem]{Remark}
\newtheorem{definition}[theorem]{Definition}
\newtheorem{claim}[theorem]{Claim}
\author[1]{Yuval Filmus}
\author[2]{Edward A. Hirsch}
\author[3]{Artur Riazanov}
\author[4]{Alexander Smal}
\author[5]{Marc Vinyals}
\affil[1]{Technion --- Israel Institute of Technology, Israel, yuvalfi@cs.technion.ac.il}
\affil[2]{Ariel University, Israel,
edwardh@ariel.ac.il}
\affil[3]{EPFL, Switzerland, tunyash@gmail.com}
\affil[4]{Technion --- Israel Institute of Technology, Israel, avsmal@gmail.com}
\affil[5]{University of Auckland, New Zealand, marc.vinyals@auckland.ac.nz}
\title{Proving Unsatisfiability with Hitting Formulas}
\begin{document}
\maketitle
\begin{abstract}
% Hitting formulas have been in use at least since \cite{Iwa89}.
% These are sets of Boolean clauses such that any two of the clauses conflict.
% %Recently, they have received additional attention \cite{Kul-many, our-papers}.

% The (un)satisfiability of a hitting formula can be easily verified by counting the number of (non-)solutions.
% This allows us to consider hitting formulas as a proof system:
% if the input formula can be weakened to a hitting formula, then it is unsatisfiable.
% To put this proof system on the map,
% it p-simulates tree-like resolution, but in the reverse direction the simulation
% is only quasi-polynomial, and there is a separation (along with the simulations,
% we prove a superpolynomial lower bound for tree-like resolution for formulas
% that have polynomial-size hitting proofs).

% Surprisingly, we did not find a strong upper bound for this system:
% the best we could do was polynomial simulation in Extended Frege
% using polynomial identity testing. As a byproduct we proved that a number
% of static (semi)algebraic systems are polynomially verifiable.

% We also consider a number of generalizations of hitting formulas 
% and prove further separation results for the respective proof systems.

A hitting formula is a set of Boolean clauses such that any two of the clauses cannot be simultaneously falsified. Hitting formulas have been studied in many different contexts at least since \cite{Iwa89} and, based on experimental evidence, Peitl and Szeider \cite{PS22} conjectured that unsatisfiable hitting formulas are among the hardest for resolution.
Using the fact that hitting formulas are easy to check for satisfiability we make them the foundation of a new static proof system {\hit}: a refutation of a CNF in {\hit} is an unsatisfiable hitting formula such that each of its clauses is a weakening of a clause of the refuted CNF. Comparing this system to resolution and other proof systems is equivalent to studying the hardness of hitting formulas.

Our first result is that {\hit} is quasi-polynomially simulated by tree-like resolution, which means that hitting formulas cannot be exponentially hard for resolution and partially refutes the conjecture of Peitl and Szeider. We show that tree-like resolution and {\hit} are quasi-polynomially separated, while for resolution, this question remains open.
For a system that is only quasi-polynomially stronger than tree-like resolution, {\hit} is surprisingly difficult to \emph{polynomially} simulate in another proof system. Using the ideas of Raz--Shpilka's polynomial identity testing for noncommutative circuits \cite{RS05} we show that {\hit} is p-simulated by {\EF}, but we conjecture that much more efficient simulations exist. As a byproduct, we show that a number
 of static (semi)algebraic systems are verifiable in deterministic polynomial time.

We consider multiple extensions of {\hit}, and in particular a proof system {\hitp} related to the {\resp} proof system for which no superpolynomial-size lower bounds are known. {\hitp} p-simulates the tree-like version of {\resp} and is at least quasi-polynomially stronger. We show that formulas expressing the non-existence of perfect matchings in the graphs $K_{n,n+2}$ are exponentially hard for {\hitp} via a reduction to the partition bound for communication complexity.

\end{abstract}
\pagebreak
\tableofcontents
%-----------------------------------------------------------------------------
\section{Introduction}
Propositional proof complexity is a well-established area with a number
of mathematically rich results. A propositional proof system \cite{CR79}
is formally a deterministic polynomial-time algorithm that verifies
candidate proofs of unsatisfiability of propositional formulas in conjunctive normal form.
The existence of a proof system that has such polynomial-size refutations for
all unsatisfiable formulas is equivalent to $\mathbf{NP}=\mathbf{co\textbf{-}NP}$,
and (dis)proving it is out of reach of the currently available methods.
Towards this goal, 
%S. A. Cook's program is the idea 
Cook and Reckhow's paper \cite{CR79}
started a program to develop new stronger proof systems 
that have short proofs for tautologies 
that are hard for known proof systems
and to prove superpolynomial lower bounds for these new systems.
The idea is that obtaining new results where our previous techniques fail helps in developing new techniques.

One of the oldest propositional proof systems is the propositional version of resolution ({\res}) \cite{Bla37, DP60}
that operates on Boolean clauses (disjunctions of literals treated as sets)
and has only a single rule that allows introducing resolvents
$\frac{\ell_1\lor\dotsb\lor\ell_k\lor x\;\quad \ell_1'\lor\dotsb\lor\ell_m'\lor\olnot{x}}{\ell_1\lor\dotsb\lor\ell_k\lor \ell_1'\lor\dotsb\lor\ell_m'}$.
Superpolynomial lower bounds on the size of a particular case of resolution proofs
are known since \cite{Tse68}, while exponential lower bounds on general {\res} proof size were proven by Haken \cite{Hak85} and Urquhart \cite{Urq87}
for the pigeonhole principle and handshaking lemma, respectively.
Furthermore {\res} encompasses CDCL algorithms for SAT \cite{BKS04,PD11}, that are the most successful SAT-solving algorithms to date.

Motivated by the quest of finding hard examples for modern SAT-solvers, Peitl and Szeider \cite{PS22} experimentally investigated the hardness of \emph{hitting formulas} for resolution. A hitting formula as a mathematical object has been studied under a number of names and in various contexts (a polynomial-time solvable SAT subclass, partitions of the Boolean cube viewed combinatorially, etc.) \cite{Iwa89, DavDav98, Kul04,Kul07,GwynneKullman13, Gwynne,KullmannZhao13,PS22,ourcombpaper}.
A formula $H = \bigwedge_i H_i$ in CNF with clauses $H_i$ is a hitting formula if every pair of clauses cannot be falsified simultaneously (that is, there is a variable that appears in the two clauses with different signs). Equivalently, the sets $S_i$ of truth assignments falsifying clauses $H_i$ are disjoint, thus in an unsatisfiable hitting formula every assignment in $\{0,1\}^n$ is covered exactly once by $S_i$'s. Peitl and Szeider conjectured that hitting formulas might be among the hardest formulas for resolution. Their conjecture was supported by experimental results for formulas with a small number of variables.

One of the reasons why hitting formulas received an abundance of attention is that they are one of the classes of CNFs that are polynomial-time tractable for satisfiability checking (along with e.g. Horn formulas and 2-CNFs) \cite{Iwa89}. First, it is straightforward to check whether a CNF formula is hitting: simply enumerate all pairs of clauses and check that they contain some variable with opposite signs. Then the number of satisfying assignments of a hitting formula is $2^n - \sum_i 2^{n-|H_i|}$, where $|H_i|$ is the number of literals in $H_i$ and $n$ is the number of variables in $H$.

This nice property allows us to think about hitting not only as a class of formulas but as an algorithm to determine satisfiability. For the algorithm to apply to any kind of formulas we need to introduce nondeterminism, and this is best modelled with a proof system. Thus we define a new static proof system based on unsatisfiable hitting formulas.
A refutation of an arbitrary CNF $F$ in the {\hit} proof system is an unsatisfiable hitting formula such that each of its clauses is a weakening of a clause in $F$ (i.e. a clause of $F$ with extra literals).

By thinking of hitting as a proof system we reinterpret the conjecture of Peitl and Szeider as the following question: is it possible to efficiently formalize the model counting argument above within the {\res} proof system?
Then the question of the hardness of hitting formulas for {\res} can be phrased in terms of the relative strength of {\res} and {\hit}: can {\hit} be separated from {\res}? That is, can we find formulas that are easy to refute in {\hit} and hard to refute in {\res}? More in general, by relating {\hit} to other proof systems we can pinpoint both the hardness of hitting formulas and the ability to formalize Iwama's counting argument in those proof systems.

It turns out that {\hit} is tightly connected to the tree-like version of {\res} ({\tres)}, which is exponentially weaker than {\res} \cite{BIW04}. It encompasses all DPLL algorithms \cite{DP60,DLL62},
which form the base of multiple (exponential-time) upper bounds for SAT (see, e.g., \cite{DH21} for a survey). A DPLL algorithm splits the input problem $F$ into subproblems $F|_{x=0}$ and $F|_{x=1}$
for some variable $x$ and applies easy simplification rules.

More precisely, {\tres} quasi-polynomially simulates {\hit} (Theorem~\ref{th:tres-qsim-hit}), and the simulation cannot be improved to a polynomial one (Theorem~\ref{th:tres-l-b}). This partially answers the question ``How hard can hitting formulas be for resolution?'' raised in \cite{PS22} in the following way. Not only every hitting formula has proofs of quasi-polynomial size, their unsatisfiability can be decided in quasi-polynomial time by a DPLL algorithm.
 The simulation also entails that every exponential-size lower bound we already have for {\tres} holds for {\hit}, which in particular allows for a separation of {\res} from {\hit}.

%\TODOJ{Introduce qp-simulation and distinguish between qp-simulation vs quasi-polynomial simulation.}

Even though the very weak proof system {\tres} is enough to quasi-polynomially simulate {\hit}, it is surprisingly difficult to push the upper bound all the way to a polynomial: even though we compare {\hit} to a number of known proof systems with different strengths with the hope of obtaining a polynomial simulation, the only system where we can polynomially simulate {\hit}
is the very powerful {\EF} (Corollary~\ref{cor:ef-sim-hit}).
As a byproduct of this result, we prove also that
various static (semi)algebraic proof systems (Nullstellensatz, Sherali--Adams, Lov{\'{a}}sz--Schrijver, Sum-of-Squares)
are indeed Cook--Reckhow (deterministically polynomial-time verifiable) proof systems even when we measure the proof size in a succinct way, ignoring the part enforcing Boolean variables. Such distinction can be safely ignored in lower bound results, but in principle ought to be accounted for when constructing upper bounds.
Efficient deterministic formal proof verification becomes more and more important because
of the increasing interest in algorithms based on sum-of-squares~\cite{BS14,FKP19}.%\TODO{anyone has a citation for SoS-based algorithms? [MV: Cited a couple of surveys]}

In more detail, we study the relation between various versions of {\hit} and known proof systems such as:
%\TODO{AR: instead of this list I would prefer these systems introduced as needed or together with their motivation and motivation for hitting and its related systems...}\MAGENTA{EH: motivation for known systems? they just existed before us.  Regarding more motivation for hitting, I would love to have it... combinatorics, perhaps?}
\begin{itemize}
\item {\resp}, defined in \cite{IS20} by analogy with
the system {\rlin} of \cite{RT08} in the same vein as
Kraj\'{\i}\v{c}ek's $\textsc{R}(\ldots)$ systems \cite{Kra98}.
No superpolynomial-size lower bound is known for it,
however, \cite{IS20} proves an exponential bound for its tree-like version.
{\resp} extends {\res} by
allowing clauses to contain affine equations modulo two
instead of just literals, and this is the weakest known bounded-depth Frege system with parity gates
where we do not know a superpolynomial-size lower bound.

We prove two separations showing that {\hit} is incomparable with {\tresp} (Sect.~\ref{sec:linear-hitting-vs-tree-like-resolution},
the separation is quasi-polynomial in one direction
and exponential in the other direction).

\item Nullstellensatz ({\ns}), defined in \cite{BeameEtAl96} (where also an exponential-size lower bound
was proved),
along with its version {\nsr} \cite{dRLNS21} that uses dual variables
($\olnot{x}=1-x$ introduced in \cite{ABSRW02}
for {\pc} \cite{CEI96}, which is a ``dynamic'' version of Hilbert's Nullstellensatz
that allows generating elements of the ideal step-by-step).
An exponential-size lower bound for {\nsr} follows from \cite{BOPCI00}
(see Corollary~\ref{cor:pebbling-ns}).
\item Cutting Planes ({\cp}), defined in \cite{CCT87}, uses linear inequalities
as its proof lines and has two rules: the rule introducing nonnegative linear combinations
and the integer rounding rule ($\frac{\sum c_ix_i\ge c}{\sum c_ix_i\ge \lceil c\rceil^{\vphantom{W}}}$ for integer $c_i$'s).
\item {\F}, defined in \cite{Rec76,CR79}, can be thought of as any
implicationally complete ``textbook'' derivation system for propositional logic.
Proving superpolynomial lower bounds for it is a long-standing open problem %in the area
that seems out of reach at the moment.
\item Systems augmented by Tseitin's extension rule and its analogues, such as {\EF}.
This rule allows the introduction of new variables denoting some functions
of already introduced variables.
\end{itemize}

Given that known proof systems do not obviously polynomially simulate {\hit}, this leaves us with the following question: does augmenting SAT algorithms with the ability to reason about hitting  formulas lead to any improvements? Or its counterpart about proof systems, how powerful are proof systems resulting from combining known proof systems with {\hit}?

Recall that a DPLL algorithm splits the input problem $F$ into subproblems $F|_{x=0}$ and $F|_{x=1}$.
Algorithms that give upper bounds for SAT use more general splittings; in fact
one can split over any tautology, that is,
consider subproblems $F\land G_1$, \ldots , $F\land G_m$,
where $G_1\lor\dotsb\lor G_k$ is a tautology.
Put in another way, one can split over an unsatisfiable formula $\olnot{G_1}\land\dotsb\land \olnot{G_k}$ ---
including unsatisfiable hitting formulas.
We use this idea, although in a DAG-like context, to introduce the following generalization of {\hit}.
\begin{description}
\item[{\hitres}] merges {\hit} with {\res}. It uses the weakening rule and also extends the main
resolution rule to
\[
\frac{C_1\lor H_1,\ \dotsc,\ C_k\lor H_k}{C_1\lor\dotsb\lor C_k}
\]
for a hitting formula $H_1\land\dotsb\land H_k$.
It is also p-simulated by {\EF} (Corollary~\ref{cor:ef-sim-hitres}).
\end{description}

Other ways in which we can generalize {\hit} while keeping with the spirit of the proof system are to allow some leeway in the requirement for the subcubes to form a partition, or in the type of objects that constitute the partition. While at first these may appear to be a mere mathematical curiosity, the connections to Nullstellensatz in the case of {\hitodd} and to the partition bound in the case of {\hitp} show that these are natural proof systems.

\begin{description}
\item[{\hitk{k}}] strengthens {\hit} by allowing to cover a falsifying assignment
with at most $k$ sets. Such proofs can be efficiently verified and p-simulated in
{\EF} using the inclusion-exclusion formula and polynomial identity testing (PIT) (Theorem~\ref{th:hitk-simulation}).
\item[{\hitodd}] strengthens {\hit} by allowing to cover a falsifying assignment
with an odd number of sets. Such proofs also can be efficiently verified 
and p-simulated in {\EF} using PIT (Prop.~\ref{prop:hitodd-verification}).
This system is equivalent to a certain version of Nullstellensatz,
which we discuss in Sect.~\ref{sec:odd-hitting}. 
We prove a lower bound for {\hitodd} (Corollary~\ref{cor:pebbling-ns})
that allows us to separate it from {\res}.
\item[{\hitp}] strengthens {\hit} by allowing the complements of affine subspaces
instead of clauses, that is, a clause can now contain affine equations instead of just literals, 
and $S_i$ is thus an affine subspace. Such proofs can be verified similarly
to {\hit} using the Gaussian elimination. We prove an exponential-size lower bound for {\hitp}
(Theorem~\ref{th:hitting-plus-lb-pm}) which, additionally, separates it from {\cp}. % \textcolor{red}{AR: CP is not defined}.
\end{description}

\tikzstyle{sline} = [line width=0.5mm]
\tikzstyle{mdashed}=[sline, dash pattern=on 2mm off 1mm]
\tikzstyle{mdotted}=[sline, dotted]
\begin{figure}[ht]
    \centering
    \begin{tikzpicture}[node distance=53mm, auto, every node/.append style={rectangle, draw=black, >=stealth }, shorten >=1mm, shorten <=1mm] 
       \node (tlres) {{\tres}};       
       \node[above left  of=tlres] (hitting) {{\hit}};
       \node[above right of=tlres] (tlresxor) {{\tresp}};
       \node[above right of=hitting] (hittingxor) {{\hitp}};
       \node[left of=hittingxor] (exfrege) {{\EF}};
       \node[right of=hittingxor] (cp) {{\cp}};
       \node[left of=tlres] (res) {{\res}};
       \node[left of=hitting] (hitodd) {{\hitodd}};
       \draw[>-, mdotted] (hittingxor) to node[above, draw=none] {\ref{th:hitting-plus-lb-pm}}  (cp);
       \draw[>->, sline] (tlres) to  (res);
       %node[below, draw=none] {\ref{th:tres-l-b}, \ref{th:res-u-b}}  (res);
       \draw[>->, mdashed] (hitting) to [bend right] node[above, draw=none] {l.b. \ref{cor:hit-vs-res}\phantom{WWWW}} %\ref{th:hitodd-vs-res}\phantom{WWWWWW}} 
        (res);
       \draw[|->, sline] (tlres) to [bend left] node[below, draw=none] {\ref{th:hit-sim-tres}, \ref{th:tres-l-b}\phantom{WWWI}} (hitting);
       \draw[>->, sline] (tlres) to [bend right] node[below, draw=none] {\phantom{WW}\cite{IS20}} (tlresxor);
       \draw[->, sline] (hitting) to [bend left] (hittingxor);
       \draw[->, sline] (tlresxor) to [bend right] node[above, draw=none] {\phantom{WWWI}as in \ref{th:hit-sim-tres}} (hittingxor);       
       \draw[>-|, mdotted] (hitting.east) to node[draw=none, below] {\ref{th:hitting-vs-treeresxor}, \ref{cor:tlresxor-vs-hitting}}(tlresxor.west);
       %\draw[->, mdashed]  (tlresxor.west) to[bend left] node[draw=none, above] {\ref{cor:tlresxor-vs-hitting}} (hitting.east);
       \draw[->, mdashed] (hitting) to[bend left] node[below, pos=0.52, draw=none] {\phantom{WWWI}\ref{th:tres-qsim-hit}} (tlres);
%       \draw[->, mdashed, draw=blue] (hittingxor) to[bend left] (hitting);
%       \draw[->, mdashed, draw=blue] (hittingxor) to[bend right] node[draw=none, rotate=315] {q.p.} (tlresxor);
       %\draw[-|, mdotted] (tlresxor) to[bend right] node[draw=none, rotate=45] {\cite{IS20}} (tlres);
       \draw[>->, sline] (hitting) to[bend left] node[left, draw=none] {\ref{cor:ef-sim-hit}} (exfrege);
       \draw[>->, sline] (hitodd) to[bend left] node[left, draw=none] {\ref{prop:hitodd-verification}} (exfrege);
       \draw[>->, sline] (hitting) to (hitodd);
       \draw[>-<, mdotted] (hitodd) to[bend right] node[left, draw=none] {\ref{prop:hitodd-tseitin}, \ref{cor:pebbling-ns}\phantom{I}} (res);
    \end{tikzpicture}
    \caption{Arrow \exarrow{->} means that $B$ p-simulates $A$, a dashed arrow \exarrow{->, mdashed} means $B$ quasi-polynomially simulates $A$. \exarrow{|-, mdotted} means a quasi-polynomial separation (a lower bound is for the system $A$). An arrowhead in the tail \exarrow{>-, mdotted} means that $A$ is exponentially separated from $B$. A dotted line \exarrow{-, mdotted} means that we do not know any simulations between $A$ and $B$. Known simulations involving {\cp} and {\EF} are not shown.}
    \label{fig:results}
\end{figure}

A summary of our simulations and separations is depicted in Figure~\ref{fig:results}, and more precise bounds are stated in Table~\ref{table:results}. Now we turn to a more detailed discussion.

\newcommand{\lb}[1]{{\textcolor{purple}{#1}}}
\begin{table}[ht]\begin{center}\begin{tabular}{ |l|c|c|c|c|c|c|c| }
\hline 
Statement&{\hit}&{\hitp}&{\small{\hitodd}}&{\tres}&{\tresp}&{\res}&{\cp}\\%hitres, ns, nsr, ef
%%&\multicolumn{2}{|c|}{substitutions}&\multicolumn{2}{|c|}{flips}\\
\hline 
\ref{th:tres-l-b}, \ref{th:res-u-b} &
$2^{\tilde O(m)}$ &
&& 
$\lb{{2^{\tilde\Omega(m^{2-\varepsilon})}}^{\phantom{V}}}$ &
&$2^{\tilde O(m)}$&\\
\ref{th:hitting-vs-treeresxor} &
$2^{\tilde O(m)}$ &
&&& 
$\lb{2^{\tilde\Omega(m^{2-\varepsilon})}}$ 
&&\\
\ref{cor:pebbling-ns} &
&&
$\lb{2^{\tilde\Omega (n)}}$ &
&&
poly&\\
\ref{cor:tlresxor-vs-hitting}, \cite{IS20} &
\lb{$2^{n^{\Omega(1)}}$} &
&&& 
poly 
&\lb{$2^{\Omega(n)}$}&\\
\ref{th:hitting-plus-lb-pm} &
&
$\lb{2^{\Omega(n)}}$ &
&&&&poly\\ 
\hline
\end{tabular}\end{center}
\caption{Precise bounds in our separations. Upper bounds are black and lower bounds are \lb{purple}.}
\label{table:results}
\end{table}

%----------------------------------------------
\subsection{Our results and methods}

%\TODO{Present old methods separately? Refer reader to somewhere?
%\begin{itemize}\item Pudlak/AD games\item xorification\item query/cert complexity\item comm complexity\end{itemize}}

\subsubsection{Simulations of {\hit}-based systems and proof verification using PIT}
Proof verification is not straightforward in static (semi)algebraic proof systems
that use either dual variables $\bar{x}=1-x$ or do not open the parentheses in $(1-x)$ 
for the negation of a variable $x$ (such as static Lov\'asz--Schrijver
or Sherali--Adams proofs or {\ns} proofs with dual variables).
A similar situation occurs with the verification of {\hit} proofs which,
contrary to most (or all?) known proof systems,
is based on model counting.
Such reasoning is not expressed naturally in propositional logic,
and it makes it difficult to simulate {\hit} proofs in other proof systems.
We observe that {\hit} proofs can be expressed similarly to {\ns} proofs
with dual variables without explicitly mentioning the side polynomials
for $x^2-x$ and $x+\bar{x}-1$ (in particular, we notice
that over $\GF(2)$, such proofs, which we call \emph{succinct {\nsr}} proofs, are equivalent
to {\hitodd} proofs, and that over any field they p-simulate {\hit} proofs in a straightforward manner).
%\TODO{Say somewhere that Hitting is simulated in succinct-NSR over *any* field.}
We show that the two problems have the same cure:
we provide an efficient polynomial identity testing
procedure for multilinear polynomials modulo $x+\bar{x}-1$ that can also be formalized in {\EF}.

Our approach uses the main idea of the Raz--Shpilka 
polynomial identity testing for noncommutative circuits \cite{RS05}. 
We introduce new variables for quadratic polynomials; crucially
it suffices to do so for a basis instead of the potentially exponential number of polynomials.
This serves as an inductive step
cutting the degrees.
Namely, at the first step we consider
two variables $x_1$ and $x_2$ and quadratic polynomials
(potentially, 
$(1-x_1)(1-x_2)$,
$(1-x_1)x_2$, 
$x_1(1-x_2)$, 
$x_1x_2$, 
$1-x_1$, 
$x_1$, 
$1-x_2$, and 
$x_2$)
appearing in the monomials $m_i$ as $\bar{x}_1\bar{x}_2$, $\bar{x}_1x_2$, etc., and replace them using linear combinations
of new variables $y^{1,2}_i$,
thus decreasing the degree by one.
At the next step we treat all the variables $y^{1,2}_i$ as a single ``layer''
(note that they are \emph{not} multiplied by each other).
We merge this layer of $y^{1,2}_i$
with $x_3$, getting a layer of variables $y^{1,2,3}_j$, and so on, 
until we reach a linear equation, which is easy to verify.

In order to implement this strategy
we prove a lemma allowing us to merge two layers of variables (Lemma~\ref{lem:merging-two-levels})
by ensuring that after the merge the equivalence of polynomials still holds.

By using this polynomial identity testing we get not only an efficient algorithm
for checking static proofs (including succinct {\nsr} proofs), 
but also a polynomial simulation in the Extended
Polynomial Calculus ({\extpc}) system that has been recently used in \cite{Ale21},
where an exponential-size lower bound has been proved.
Given that {\extpc} over $\GF(2)$ is equivalent to {\EF} (Prop.~\ref{prop:ef-sim-extpc}),
we obtain p-simulations of {\hit}, {\hitres}, {\hitodd} and {\hitk{k}} in {\EF}.

A simpler proof that succinct {\sa} is verifiable in polynomial time
was developed independently in~\cite{dRPR24}, where the proof system
is named \emph{semantic} {\sa}. Interestingly, after applying
algebraic manipulations, their proof implicitly reduces the problem of
verifying a {\sa} proof to that of verifying a {\hit} proof. This
suggests that proofs systems relying on model counting are not that
uncommon after all.

\subsubsection{Separations of {\hit} from classical systems}
A polynomial simulation of {\tres} in {\hit} (Theorem~\ref{th:hit-sim-tres}) can
be easily shown by converting {\tres} to a decision tree, then the assignments in the leaves provide a disjoint partition of the Boolean cube.
We show a quasi-polynomial simulation in the other direction
through careful analysis
of a recursive argument (Theorem~\ref{th:tres-qsim-hit}).
The main idea is that an unsatisfiable formula containing $m$ clauses must necessarily
contain a clause of width $w \le \log_2 m$, and in a hitting formula this clause must contain a variable
that occurs with the opposite sign in at least $(m-1)/w$ clauses.
Making a decision over this variable thus removes a lot of clauses in one of the two branches.
We also employ a generalization of this idea to show that {\hitk{k}} proofs
can be quasi-polynomially simulated in {\hit} (Prop.~\ref{prop:hitk-qsim-hit}) and hence in {\tres}.

A polynomial simulation in the other direction is impossible because of a superpolynomial separation.
To show this result (Theorem~\ref{th:tres-l-b}) we use query complexity, and in particular,
the result of~\cite{AKK16} separating unambiguous query complexity from randomized query complexity.
We lift it using xorification to obtain the desired separation.

We then obtain a two-way separation between {\hit} and {\tresp} (Sect.~\ref{sec:hitting-vs-res-ns}). On the one hand Tseitin formulas
are hard for {\res} \cite{Urq87} and hence for {\hit}. On the other hand, 
\cite{IS20} shows that they have polynomial-size {\tresp} (and thus also {\hitp}) proofs.
In the other direction, similarly to the separation between {\hit} and {\tres},
we again use the separation of \cite{AKK16} between unambiguous certificate
complexity and randomized query complexity as our starting point. However, since for {\tresp} we are unable
to use decision trees, we need to go through randomized communication complexity arguments, using a randomized query-to-communication lifting theorem~\cite{GPW17}.
%As a lifting gadget, we use the indexing function.

Eventually, we discuss separations of {\hit} from {\res} and {\ns}.
While the relevant lower bounds for {\hit} follow directly from known lower bounds for {\tres},
the other direction seems much more difficult, if possible at all.
% MV: It is true over any field: a polynomial equality remains valid modulo a polynomial.
One natural candidate for a separation result could be
the formulas that we used to separate
{\hit} from {\tres}, but this cannot work because they turn out to have {\res} proofs of polynomial size
(Theorem~\ref{th:res-u-b}).
%\TODO{move the proof to the appendix and mention this fact here?}
We show this fact using dag-like query complexity~\cite{GGKS20}, the analogue of resolution width in query complexity, which stems from a game characterization of {\res} \cite{Pud00,AD08}.
We need to reprove the result of~\cite{AKK16} accordingly, improving it to a separation between unambiguous dag-like query complexity and randomized query complexity. This immediately yields an upper bound on the {\res} width.
%\TODO{more to say right here?\\MV: I would not go into more details}
Concerning {\ns}, it is a simple observation that {\hit} is simulated by {\ns} with respect to width vs degree.
Furthermore, as we discussed above, succinct {\nsr} proofs (over any field)
simulate {\hit} with respect to size, therefore separating {\hit} from {\res} would amount to separating explicit vs succinct {\nsr} size.

\subsubsection{A lower bound for {\hitodd}}

As mentioned above {\hitodd} is polynomially equivalent to succinct {\nsr} proofs over $\GF(2)$, and
we explain this in more detail in the beginning of Sect.~\ref{sec:odd-hitting}. It is
easy to see that {\hitodd} has short proofs of Tseitin formulas and thus it is
exponentially separated from {\res}.
The opposite direction (Cor.~\ref{cor:pebbling-ns}) requires slightly more effort.
It is known that {\res} width can be separated from {\ns} degree~\cite{BOPCI00}.
We use this result to get our size separation
using xorification and the random restriction technique of Aleknhovich
and Razborov (see~\cite{Ben02}).

\subsubsection{A lower bound for {\hit}(\texorpdfstring{$\oplus$}{⊕})}

Our lower bound for {\hitp} (Theorem~\ref{th:hitting-plus-lb-pm}) uses a communication complexity argument. Communication complexity reductions have a long history of applications in proof complexity \cite{BPS05,HN12,GP18,IS20,dRNV16}. The first step in these reductions is a simulation theorem, which shows that a refutation of an arbitrary CNF $\phi$ in the proof system of interest can be used to obtain a low-cost communication protocol solving the communication problem $\Search(\phi)$: given an assignment to the variables of $\phi$, find a clause of $\phi$ falsified by this assignment. The second step is reducing a known hard communication problem (usually set disjointness) to $\Search(\phi)$ for a carefully chosen CNF $\phi$.

Until recently the applications of these reductions were limited to either proving a lower bound for a tree-like version of the system or proving a size-space tradeoff, neither of which applies to our result. However, over the last few years, the list of applications of the communication approach in proof complexity has grown significantly. A major breakthrough came in \cite{Sok17,GGKS20} with a dag-like lifting theorem from resolution to monotone circuits and cutting plane refutations. Another novel idea was introduced in \cite{GHJMPRT22}, where the authors derived a lower bound for Nullstellensatz via a communication-like reduction from the $\Omega(\sqrt{n})$ lower bound on the approximate polynomial degree of $\mathrm{AND}_n$ \cite{NS95}.

We use yet another twist on this idea: we apply a communication reduction to the \emph{partition bound} \cite{JK10}, a generalization of randomized communication protocols which simulates {\hitp} (Lemma~\ref{lem:hitting-xor-reduction}). To the best of our knowledge this is the first application of the partition bound in a proof complexity context.
We then adapt (Theorem~\ref{th:perfect-matching-reduction}) a communication reduction from set disjointness in~\cite{IR21} so that it works for the partition bound and use the fact that set disjointness is still hard for the partition bound to get our lower bound (Theorem~\ref{th:hitting-plus-lb-pm}). The choice of the reduction of~\cite{IR21} is not particularly important, and we believe that reductions from~\cite{BPS05,GP18,IS20} should also work. A nice feature of the reduction we use is that we get a lower bound for a natural combinatorial principle: a formula encoding the non-existence of a perfect matching in a complete bipartite graph $K_{n,n+2}$. Because this formula is known to have short {\cp} proofs, we obtain a separation between {\hitp} and {\cp} as an immediate corollary.

\subsection{Further research}

\paragraph{Relation between {\hit} and {\res}.} 
%\TODO{\res vs \hit: note that it is enough to refute irreducible formulas} 
%\MAGENTA{YF: Relate to complexity measures of functions.}
% Theorem~\ref{th:tres-qsim-hit} shows that {\tres} quasi-polynomially simulates {\hit}, and hence {\res} quasi-polynomially  simulates {\hit} as well. Can this be improved? 
% Due to Theorem~\ref{th:extpc-sim-hit}, we know that {\extpc} p-simulates {\hit}. 
% Can we show p-simulation for weaker proof systems, particularly for {\res}? 
% %Corollary~\ref{cor:hit-vs-res} states that formul {\hit} and {\res} are exponentially separated.
% Is it possible to show this separation directly by separating corresponding complexity 
% measures, $\dagw$ and $\ucert$, and use an argument similar to the proof of Theorem~\ref{th:tres-l-b}?
% (Note that we cannot use the same argument due to Theorem~\ref{th:res-u-b}.)
Although we have gained a lot of understanding of the hardness of hitting formulas for resolution, the initial question of Peitl and Szeider is not fully answered. In particular, we do not know whether hitting formulas can be superpolynomially hard for {\res}. The negative answer implies a simulation of {\hit} by {\res}. 
To show the positive answer it is sufficient to separate two query complexity measures: dag-like query complexity ($\dagw$) and unambiguous certificate complexity ($\ucert$).
%For the falsified clause search problem, dag-like query complexity corresponds to the resolution width, unambiguous certificate complexity corresponds to the {\hit} refutation width.
The dag-like query complexity of the falsified clause search problem for a formula $F$ corresponds to the resolution width of $F$.
The unambiguous certificate complexity for this problem corresponds to the width of {\hit} refutations of $F$.
Note that unambiguous certificate complexity only makes sense for functions, while dag-like query complexity is defined for (total) relations.
 % C <= UC <= w
Unfortunately, separating even regular certificate complexity ($\cert$) and $\dagw$ is an open problem for \emph{functions} (without the uniqueness requirement the certificate complexity can only decrease, so it might be easier to separate $\dagw$ from $\cert$ than from $\ucert$). Lemma~\ref{lem:composition-and-width} and Lemma~\ref{lem:width-cheatsheet} show that $\dagw$ is resistant to known lower bound techniques in the field of query complexity, so tackling it will likely lead to finding new techniques there. Notice that we know how to separate $\dagw$ and $\cert$ for \emph{relations}: every lower bound on the resolution width for an $O(1)$-CNF formula constitutes a separation for the corresponding falsified clause search problem. Such a separation (constant vs.\ polynomial) is unachievable for functions (we cannot hope for better than quadratic separation for functions as $\dagw(R) \le \cert(R)^2$). Can we use ideas from resolution lower bounds to separate $\dagw$ and $\ucert$ (or at least $\cert$)?

%\TODO{HITTING(2) vs HITTING(1).} \MAGENTA{YF: Relate to complexity measures of functions.}
\paragraph{Separate {\hit} and {\hitk{2}}.} With xorification like in Lemma~\ref{lem:xor_dt} this problem can be shown to be equivalent to a simple (if only in the statement!) question in query complexity: separate unambiguous certificate complexity and $2$-unambiguous certificate complexity (where every input has one or two certificates). It is known how to separate one-sided versions of these query models \cite{GKY22}, but similarly to the question of {\hit} vs {\res} it is unclear how to extend this to the two-sided case. 

%\TODO{tree-like LINEAR HITTING vs {\hitp}.} 
\paragraph{Is it possible to separate {\hitp} and {\tresp}?} 
In Section~\ref{sec:no-qp-simulation} we give evidence that a simulation of {\hitp} by {\tresp} along the lines of Theorem~\ref{th:tres-qsim-hit} is not possible. That, however, does not rule out the existence of such a simulation.  \cite[Conjecture~5.1.3]{SS21} conjectures that every affine subspace partition can be refined to one corresponding to a parity decision tree with a quasi-polynomial blow-up. With some caveats\footnote{The refinement might be non-constructive, but its mere existence does not imply the simulation. The simulation might produce parity decision trees that are not refinements of the initial {\hitp} refutation but nevertheless, solve the relation $\Search(\phi)$.}, the statement of this conjecture is equivalent to the existence of  quasi-polynomial simulation of {\hitp} by {\tresp}. So, is there an exponential separation between these two systems? It seems that communication-based lower bounds for {\tresp} can be transferred to {\hitp} as it is done in Section~\ref{sec:hitp-lowerbound}. There are several other techniques that yield {\tresp} lower bounds such as prover-delayer games \cite{IS20,Gry19}, reduction to polynomial calculus degree \cite{GK18}, and the recent lifting from decision tree depth to parity decision tree size directly \cite{CMSS22,BK22}. None of those seem to work for {\hitp}, so it is reasonable to think that some of the yielded formulas may have an upper bound in {\hitp}. The most promising technique seems to be the lifting since it yields a wide family of formulas hard for {\tresp} with the source of hardness inherent to the tree-like structure of refutations.

%\RED{What is tree-like hitting?!}
%\MAGENTA{YF: It is the proof system corresponding to polynomial size parity decision trees. Recall the tree-like Resolution qp-simulates Hitting. We cannot show the corresponding result for Hitting(+), even though the corresponding query complexity result does hold.}

%\TODO{What lifting from dt to PDTsize \cite{CMSS22} means in the context of separating {\hitp} form {\tresp}?}

\paragraph{Better upper bound on {\hit}.} One intriguing matter is that although a very weak proof system such as {\tres} is enough to quasi-polynomially simulate {\hit}, we need to go all the way to the very strong proof system {\EF} for the simulation to become polynomial. A natural question is then what is the weakest proof system that is enough to polynomially simulate {\hit}.

It is consistent with our findings that a fairly weak proof system such as {\nsr} is already enough to simulate {\hit}; indeed this would be the case if {\nsr} and succinct {\nsr} were equivalent. Hence we ask the same question regarding succinct (semi)algebraic proof systems: what is the weakest proof system that polynomially simulates succinct {\nsr} or succinct {\sa}? And in particular, is succinct {\nsr} equivalent to {\nsr} and is succinct {\sa} equivalent to {\sa}? One way to answer all these questions would be to formalize the PIT of Theorem~\ref{th:extpc-sim-hit} in a weaker proof system.

%[Cue to no known simulation of hitting($\oplus$)...] \TODO{Upper bound on {\hitp}.} 
%\RED{ (Why can't we do it in {\extpc}? Just introduce new vars.)} 
%\MAGENTA{YF: Doesn't work, since we are relying on a noncommutative PIT.} \
%\RED{will check, but at least Odd Hitting should be OK}

The situation with {\hitp} is even worse. We have shown how to p-simulate most of the generalizations of {\hit} that we defined, including {\hitodd} and {\hitk{k}}, in {\extpc}, but the argument does not work as is for {\hitp} since we are relying on a noncommutative PIT. Therefore we do not know even an {\EF} simulation of {\hitp} (though it is of course quite expected). %\TODOJ{sure there is no trick? very strange}

\paragraph{Non-automatability of {\hit}.} It follows from Theorem~\ref{th:tres-qsim-hit} and quasi-polynomial automatability of {\tres} \cite{BP96} that {\hit} is also quasi-polynomially automatable. Can we show that it is impossible to do better? We think that it is possible to adapt the similar result of de Rezende \cite{dR21} for {\tres}.

%\TODO{semialgebraic: refer to Albert's ALN16 as a survey (is there more than one?)}
%=============================================================================
\section{Basic definitions}
%-----------------------------------------------------------------------------
\subsection{Basic notation}

For a function $f:\mathbb{N}\to\mathbb{R}$, $\tilde O(f)$ and $\tilde\Omega(f)$ denote $O$ and $\Omega$ up to logarithmic factors,
that is, $g = \tilde O(f)$ and $h = \tilde \Omega(f)$ if $g = O(f\log^Cf)$ and $h = O(f/\log^Cf)$ respectively for a constant $C$.
For example, $2^nn^2=\tilde O(2^n)$ and $n/\log n=\tilde\Omega(n)$.
%\TODO{Is there a more formal definition, given that we use it both for exponents and
%logarithms?}
%\RED{YF: $g = \tilde O(f)$ if $g = O(f\log^Cf)$ for some $C$. Sometimes we also allow $g = O(f^{1+o(1)})$, but I'm not sure this is the case here.}

Let $f\colon\{0,1\}^n\to\{0,1\}$ be a Boolean function. The \emph{deterministic query complexity of $f$}, denoted by $\D(f)$, is the minimal number of (adaptive) queries to the input variables that is enough to compute $f(x)$ for any input $x$.
The \emph{randomized query complexity of $f$}, $\R(f)$, is 
the minimum number of queries needed by a randomized algorithm
that outputs $f(x)$ for any input $x$ with probability at least $2/3$.
A partial assignment $\alpha\in\{0,1,{*}\}^n$ is a \emph{certificate} for $f$ if for any two assignments $x,y\in\{0,1\}^n$ agreeing with $\alpha$, $f(x) = f(y)$. The \emph{size} of a certificate is the number of non-star entries.
The \emph{certificate complexity of $f$ on an input $x$}, denoted $\cert(f,x)$, is size of the smallest certificate $\alpha$ such that $x$ agrees with $\alpha$. 
For $b\in\{0,1\}$, the \emph{(one-sided) $b$-certificate complexity of $f$} is defined as $\cert_b(f) = \max_{x: f(x) = b} \cert(f,x)$. 
The \emph{(two-sided) certificate complexity of $f$} is the maximum of $0$- and $1$-certificate complexities, $\cert(f) = \max\{\cert_0(f),\cert_1(f)\}$. 
We say that a family of certificates $A\subset \{0,1,*\}^n$ is \emph{unambiguous} if any two distinct certificates $\alpha,\beta\in A$ conflict, i.e., there is no assignment that agrees with both $\alpha$ and $\beta$. 
For $b\in\{0,1\}$, the \emph{(one-sided) unambiguous $b$-certificate complexity of $f$}, $\ucert_b(f)$, is the minimum number $w$ such that there is an unambiguous family of certificates $A$ such that $A$ contains only certificates of size at most $w$ and every $x\in f^{-1}(b)$ agrees with some certificate in $A$. 
The \emph{(two-sided) unambiguous certificate complexity of $f$} is defined as $\ucert(f) = \max\{\ucert_0(f), \ucert_1(f)\}$.

\TODOJ{New paragraph below --- please check.}
For the definition of the basic proof complexity notions such as proof system and p-simulation, 
we refer the reader to \cite{CR79}. We consider also quasi-polynomial simulations:
a proof system $\Pi_1$ \emph{quasi-polynomially simulates} proof system $\Pi_2$
if for certain $k\in\mathbb{N}$,
for every formula $F$, the system $\Pi_1$ has a proof of $F$ of size at most $2^{(\log s)^k}$,
where $s$ is the size of the shortest proof of $F$ in $\Pi_2$.
One could define (and name) a constructive (analogous to p-simulation) 
version of this notion, and in fact our quasi-polynomial simulations 
(Theorem~\ref{th:tres-qsim-hit}, Proposition~\ref{prop:hitk-qsim-hit})
are constructive, that is, we can produce the proofs in the simulating proof system
in time polynomial in their length. This is not important for the separation
results, so we use the term ``quasi-polynomial simulation'' without emphasizing
the constructiveness.
We say that a proof system $\Pi_1$ is \emph{quasi-polynomially separated}
from $\Pi_2$ if there is an infinite sequence of formulas $F_n$, whose size tends to infinity, and a specific $k\in\mathbb{N}$ such that
$\Pi_2$ has no proofs of $F_n$ of size $2^{(\log (s_n+|F_n|))^k}$,
where $s_n$ is the size of the shortest proof of $F_n$ in $\Pi_2$.
Somewhat abusing the notation we say that $\Pi_1$ is \emph{quasi-polynomially stronger}
than $\Pi_2$ if it polynomially simulates $\Pi_2$ and is quasi-polynomially separated from it.

We use the following notation for widely known proof systems:
{\res} for Resolution,
{\tres} for tree-like Resolution,
{\resp} for Resolution over XORs of \cite{IS20},
{\tresp} for its tree-like version,
{\cp} for Cutting Planes,
{\ns} for Nullstellensatz,
{\pc} for Polynomial Calculus,
{\F} for Frege and {\EF} for Extended Frege.

\emph{Deterministic communication complexity} of a search problem defined by a ternary relation $R$
is the minimal amount of communication (number of bits) that is enough to solve the following 
communication problem for two players on any input: Alice is given $x$, Bob is given $y$, 
and their goal is to find some $z$ such that $(x,y,z) \in R$.
Alice and Bob can exchange information by sending bit messages to each other. 
At the end of the game both players must know $z$.
\emph{(Public coin) $\varepsilon$-error randomized communication complexity} of a search problem
is the minimal amount of communication that is enough for players to win the communication game 
with probability at least $1-\varepsilon$ if the players have access to a public source of random bits.
If $\varepsilon$ is not explicitly specified then we assume $\varepsilon = 1/3$.
More information on the standard definitions  of communication complexity can be found in~\cite{KN97}.

%-----------------------------------------------------------------------------
\subsection{Hitting formulas and proof system}
Iwama \cite{Iwa89} started to study hitting formulas as a polynomial-time tractable subclass
of satisfiability problems (see also \cite{Kul04}).
\begin{definition}[Hitting formula]\label{def:hitformula}
A \emph{hitting formula} is a formula $F=C_1\land\dotsb\land C_m$ 
in conjunctive normal form such that every two of its clauses $C_i$ and $C_j$
contain contrary literals, that is, there is some literal $\ell$ such that $\ell\in C_i$
and $\bar{\ell}\in C_j$; in other words, $C_i\lor C_j$ is a tautology.
\end{definition}

Sometimes the notion is defined for formulas in disjunctive normal form. 
We call them a different name to avoid misunderstanding.

\begin{definition}[Unambiguous DNF]\label{def:unambiguous-DNF}
An \emph{unambiguous DNF} is the negation of a hitting formula,
that is, every two its terms (conjunctions) contradict each other.
\end{definition}

\begin{definition}[{\hit} proof system]\label{def:hit}
A refutation of a CNF $F$ in {\hit} is an unsatisfiable hitting formula $H$
such that every clause $C$ in $H$ has a strengthening $C'\subseteq C$ in $F$.
\end{definition}

{\hit} refutations can be verified in polynomial time:
the unsatisfiability of $H$ can be easily checked by counting
the number of falsifying assignments, as implicitly noticed by Iwama \cite{Iwa89}
(note that the sets of falsifying assignments
for any two clauses of $H$ are disjoint), and matching clauses
to their strengthening is done simply by considering all pairs $C\in H$, $C'\in F$.

The soundness of {\hit} is trivial, the completeness is given by the ``complete'' 
hitting formula consisting of all possible clauses containing all the variables of $F$:
the unique assignment falsifying such a clause $C$ must also falsify some clause $C'$ of
(unsatisfiable) $F$, which is then the required strengthening of $C$.

%-----------------------------------------------------------------------------
\subsection{Other {\hit}-based proof systems}
\subsubsection{{\hitres}}
{\hit} is a ``static'' proof system with no real derivation procedure.
We add more power to it by incorporating such steps into a {\res} refutation.
Indeed, a resolution step can be generalized to resolve over any contradiction, 
not just $x\land \bar{x}$.
In {\hitres} we resolve by hitting formulas:
\begin{definition}[{\hitres}]\label{def:hitres}
This proof system embraces both {\hit} and {\res}.
One derivation step uses an unsatisfiable hitting formula
$H_1\land\dotsb\land H_k$:
\[
\frac{C_1\lor H_1,\ \dotsc,\ C_k\lor H_k}{C_1\lor\dotsb\lor C_k}.
\]
We also allow weakening steps:
\[
\frac{C}{C\lor D}.
\]
\end{definition}
%\TODO{Can't we get rid of weakening, as in {\res}?} \RED{YF: It's unclear, since ``strenthening'' might make a hitting formula non-disjoint.}
\begin{proposition}
{\hitres} p-simulates both {\hit} and {\res}.
\end{proposition}
\begin{proof}
{\hitres} generalizes {\res}: if one uses the hitting formula $x \land \bar{x}$
at every step, {\hitres} turns exactly into {\res}.
On the other hand, in {\hit} we need to demonstrate that every clause of a hitting formula
is a weakening of some input clause, and this can be simulated using the weakening rule.
\end{proof}
%-----------------------------------------------------------------------------
\subsubsection{{\hitodd}}
While a hitting formula covers every falsifying assignment exactly once,
that is, it satisfies exactly one clause, an odd hitting formula
does this an odd number of times.

\begin{definition}[Odd hitting formula]\label{def:hitoddformula}
An \emph{odd hitting formula} is a formula $F=C_1\land\dotsb\land C_m$ 
in conjunctive normal form such that every falsifying assignment
falsifies an odd number of its clauses.
\end{definition}
\begin{definition}[{\hitodd} proof system]\label{def:hitodd}
A refutation of a CNF $F$ in {\hitodd} is an unsatisfiable odd hitting formula $H$
such that every clause $C$ in $H$ has a strengthening $C'\subseteq C$ in $F$.
\end{definition}

It is not straightforward how to verify that a (not necessarily unsatisfiable) formula is an odd hitting formula, and  how to verify that a formula is an unsatisfiable odd hitting formula (thus verifying {\hitodd} proofs). We show it in~Prop.~\ref{prop:hitodd-formula-verification} and  Prop.~\ref{prop:hitodd-verification}.
%Note that every formula (except for a trivial one) has falsifying assignments
%-----------------------------------------------------------------------------
\subsubsection{\texorpdfstring{\hitk{k}}{Hitting[k]}}

One can generalize hitting formulas by allowing
a falsifying assignment to falsify a limited number
of clauses (and not just a single clause) \cite{Kul07}.
\begin{definition}[Hitting-$k$ formula]\label{def:hitk-formula}
A \emph{hitting-$k$ formula} is a formula $F$ 
in conjunctive normal form such that every 
assignment falsifying $F$
falsifies at most $k$ clauses of $F$.
\end{definition}
%The condition in this definition can be verified by checking
%subsets of $F$ of size $k-1$.
\begin{definition}[\hitk{k}]\label{def:hitk}
A refutation of a CNF $F$ in {\hitk{k}} is an unsatisfiable hitting-$k$ formula $H$
such that every clause $C$ in $H$ has a strengthening $C'\subseteq C$ in $F$.
\end{definition}
We show in Theorem~\ref{th:hitk-simulation} that {\hitk{k}} refutations can be verified in polynomial time.

%-----------------------------------------------------------------------------
\subsubsection{\texorpdfstring{\hitp}{Hitting(xor)}}

{\hitp} stands to {\hit} the same way as {\resp} stands to {\res},
where {\resp} is the system defined in \cite{IS20} that
allows clauses to contain affine equations modulo two
instead of just literals.
It resembles the system {\rlin} of \cite{RT08} and falls under the concept
of Kraj\'{\i}\v{c}ek's $\textsc{R}(\ldots)$ systems \cite{Kra98}.

\begin{definition}[Hitting$(\oplus)$ formula]
A \emph{hitting$(\oplus)$ formula} decomposes $\{0,1\}^n$ into disjoint affine subspaces over $\GF(2)$. Namely, it is a conjunction of $\oplus$-clauses of the form
$\bigvee_k \left(c_k\oplus\bigoplus_{i\in I_k} x_i\right)$, where $c_k\in\{0,1\}$ are constants,
$x_i$'s are variables, and any two its $\oplus$-clauses do not share
a common falsifying assignment.
\end{definition}

Note that we can check that two affine subspaces are disjoint using Gaussian elimination, and this gives an efficient way of checking whether a given formula is hitting$(\oplus)$.

$\oplus$-clauses can be thought of as sets of linear (affine) equations similarly to
clauses that we usually think of as sets of literals.

\begin{definition}[{\hitp} proof system]\label{def:hitp}
A refutation of a CNF $F$ in {\hitp} is an unsatisfiable hitting$(\oplus)$ formula $H$
such that every $\oplus$-clause $C$ in $H$ has a strengthening $C'\subseteq C$ in $F$.
\end{definition}

Note that {\hitp} can be thought of also as a \emph{proof system
for sets of affine subspaces covering $\{0,1\}^n$}, that is, 
unsatisfiable systems of disjunctions of linear (affine) equations.

%=============================================================================
\section{PIT helps to simulate {\hit}, and more}\label{sec:verification}

\TODOJ{This section was moved from 4 to 3.}
%\TODOJ{MV: Should we move this section before Section 3? That is the order we use in the introduction (which I like). Maybe after the deadline, though. // EH: I fully agree, and indeed after the deadline (along with formal definition and treatment of succinct NSR).}

%-----------------------------------------------------------------------------
\subsection{{\EF} p-simulates {\hit}}

We prove that {\hit} can be p-simulated at least in the most powerful logical
propositional proof system, {\EF}.
The obstacle is that the soundness of {\hit} is based on the counting argument
that involves the number of assignments falsified by a clause, and it is not easy
to express this argument in propositional logic.

Our strategy is to p-simulate {\hit} in a strong algebraic system
that is p-equivalent to {\EF} in the case of $\GF(2)$.

%---> intro ???
There are several proof systems extending the power of {\pc} by
allowing to express polynomials in a more compact way than linear combinations of monomials.
Grigoriev and Hirsch~\cite{GH03} introduced {\fpc} that allows
to express polynomials as algebraic formulas without opening the parentheses.
Of course, this needs usual associativity--commutativity--distributivity rules
to transform these formulas. The next powerful system is
{\extpc} considered by Alekseev \cite{Ale21}. This is simply {\pc}
with Tseitin's extension rule generalized so that variables
can be introduced for arbitrary polynomials. It can be viewed as
a way to express {\pc} proofs where polynomials can be represented
as algebraic circuits (but transformations of these circuits 
must be justified using the definitions of extension variables 
that denote gates values).
Eventually, Grochow and Pitassi~\cite{Pit96,GP18} suggested to
generalize proof systems to allow the randomized verification of the proofs,
and in these proof systems, one can switch for free between different 
circuit representations of a polynomial.

A {\F} system \cite[\S2]{CR79} is defined as any implicationally complete inference system
that uses sound constant-size rule schemata for Boolean formulas (a schema
means that the formulas in the rules are represented by meta-variables,
for example, $F$ and $G$ in the modus ponens rule $\frac{F;\ F\supset G}{G}$
can be any formulas). An {\EF} system additionally allows us to introduce new
variables using the axiom schema $x \Leftrightarrow A$ for \emph{any} formula $A$, 
where $x$ is a new variable.

Grigoriev and Hirsch~\cite[Theorem 3]{GH03} prove that {\fpc} (over any field),
a system that allows us to represent polynomials using arbitrary algebraic formulas and to transform them using the ring rules,
p-simulates {\F} (and also a similar statement for constant-depth {\fpc} over finite fields versus {\F} with modular gates). They also state that {\F} p-simulates {\fpc} over $\GF(2)$ \cite[Remark 5]{GH03}. We include a formal proof of this statement for completeness.
Namely, we prove that {\fpc} over $\GF(2)$ is a {\F} system itself
(and it is known that all sound and implicationally complete {\F} systems over all possible sets of Boolean connectives are equivalent \cite[Theorem 5.3.1.4.i]{Rec76}).
\begin{proposition}\label{prop:f-sim-fpc}
{\fpc} over $\GF(2)$ is a {\F} system.
\end{proposition}
\begin{proof}
%In fact, {\fpc} over $\GF(2)$ \emph{is} a {\F} system:
{\fpc} operates with polynomial equations over $\GF(2)$,
and these equations can be considered 
as Boolean formulas that use $\oplus$, $\land$ and the negation.
All its rules are, of course, sound, the system is complete, 
and the implicational completeness
can be shown as follows: if $A_1,\dotsc,A_k\models F$, then
$A_1,\dotsc,A_k,1-F\models 1$; by completeness,
there is a derivation $A_1,\dotsc,A_k,1-F\vdash^* 1$,
which we can multiply by $F$ \cite[Remark 2]{GH03}.
\end{proof}

\begin{definition}[\cite{Ale21}]\label{def:extpc}
An {\extpc} refutation over $R$ of a set of polynomials $P\subset R[x_1,\dotsc,x_n]$
is a {\pc} refutation over $R$ of a set $P\cup Q$,
where $Q$ consists of polynomials defining new variables $y_i$: 
%\begin{multline*}
\[
Q := \{y_1 - q_1(x_1, \dotsc, x_n), y_2 - q_2(x_1, \dotsc, x_n, y_1), \dotsc, 
%\\ 
y_m - q_m(x_1, \dotsc, x_n, y_1, \dotsc, y_{m - 1})\}
\]
%\end{multline*}
where $q_i \in R[x_1,\dotsc,x_n, y_1, \dotsc, y_{i - 1}]$ are arbitrary polynomials. 
\end{definition}
While \cite{Ale21} defines {\extpc} over arbitrary fields and even rings, 
we use it over finite fields only.

Similarly to Prop.~\ref{prop:f-sim-fpc}, we prove that {\extpc} over $\GF(2)$ is an {\EF} system
(and it is known that
all {\EF} systems are p-equivalent \cite[Theorem 5.3.2.a]{Rec76}).
\begin{proposition}\label{prop:ef-sim-extpc}
{\extpc} over $\GF(2)$ is an {\EF} system.
\end{proposition}
\begin{proof}
Extension variables can be introduced for any polynomial,
but again these polynomials are Boolean formulas in the basis of $\{\oplus,\land,\bar{\phantom{x}}\}$.
So {\extpc} is an {\EF} system.
\end{proof}

%--- 

The main theorem of this section is
\begin{theorem}\label{th:extpc-sim-hit}
{\extpc} over a finite field p-simulates {\hit}.
\end{theorem}
We prove it in the next subsection.

\begin{corollary}\label{cor:ef-sim-hit}
{\EF} p-simulates {\hit}.
\end{corollary}
\begin{proof}
Follows from Theorem~\ref{th:extpc-sim-hit} and Prop.~\ref{prop:ef-sim-extpc}.
\end{proof}

\begin{corollary}\label{cor:ef-sim-hitres}
{\EF} p-simulates {\hitres}.
\end{corollary}
\begin{proof}
We show how {\EF} simulates a single step of {\hitres} refutation that uses a hitting formula $H_1\land\dotsb\land H_k$:
\[
\frac{C_1\lor H_1,\ \dotsc,\ C_k\lor H_k}{C_1\lor\dotsb\lor C_k}.
\]
Since {\EF} is p-equivalent to {\extres}, one can construct in polynomial time an {\extres} refutation $H_1,\dotsc,H_k\vdash^* \bot$ by Corollary~\ref{cor:ef-sim-hit}.
Observe that if we weaken the premises to $C_1\lor H_1$, \ldots, $C_k\lor H_k$,
then this refutation turns into a derivation of a subset of $C_1\lor\dotsb\lor C_k$.
One can now combine the simulations of all steps into a single {\extres} refutation.
\end{proof}

%-----------------------------------------------------------------------------
\subsection{Proof of Theorem~\ref{th:extpc-sim-hit}: {\extpc} p-simulates {\hit}}
We have a hitting formula $H=\{C_1,\dotsc,C_m\}$, translate it into a system of polynomial
equations $\{h_i=0\}$ and want to construct an {\extpc} refutation of this system. 
In fact, $\sum_ih_i \equiv 1$ as polynomials (in what follows, we use the notation
$\equiv$ for the equality of polynomials). 
This is certainly true pointwise on $\{0,1\}^n$, 
these polynomials are multilinear,
and thus $\sum_ih_i$ is identical to~$1$.
It remains to derive this fact in {\extpc}.

We translate formulas in CNF to systems of polynomial equations
using the dual variables as in {\pcr} of \cite{ABSRW02}: for every variable $x$,
we introduce a variable $x^R$ along with the axiom $x+x^R-1=0$.
Thus every clause $C_i=\ell_1\lor\dotsb\lor \ell_k$ 
is represented by a monomial $m_i=\ell_1\dotsb\ell_k$:
every negative literal of $C_i$ is translated to its variable,
and every positive literal is translated to the dual variable.
In the proof below, we ignore these formalities and speak
in the terms of $x$ and $1-x$ instead of $x$ and $x^R$.
We switch between these two representations ($x^R$ and $1-x$)
locally when needed (in particular, we never switch to the $1-x$
representation for more than two variables in a monomial,
and switch back to $x^R$ as soon as we are done with the respective step).

As mentioned in the introduction, our approach is based on the Raz--Shpilka deterministic polynomial identity testing
for noncommutative circuits~\cite{RS05}. 
The main idea is to introduce new variables for quadratic polynomials:
it suffices to do it for a basis. Namely, at the first step we consider
two variables $x_1$ and $x_2$ and quadratic polynomials
(potentially, 
$(1-x_1)(1-x_2)$,
$(1-x_1)x_2$, 
$x_1(1-x_2)$, 
$x_1x_2$, 
$1-x_1$, 
$x_1$, 
$1-x_2$, and 
$x_2$)
appearing in the monomials $m_i$ as $x_1^Rx_2^R$, $x_1^Rx_2$, etc., and replace them using linear combinations
of new variables $y^{1,2}_i$,
thus decreasing the degree by one.
At the next steps we treat all the variables $y^{1,2}_i$ as a single ``layer'' 
(note that they are \emph{not} multiplied by each other).
We merge this layer of $y^{1,2}_i$
with $x_3$, getting a layer of variables $y^{1,2,3}_j$, and so on, 
until we reach a linear equation, which is easy to verify.

In order to implement this strategy,
we prove a lemma allowing to merge two layers of variables.
This lemma holds over any field $\mathbb{F}$.
\begin{lemma} \label{lem:merging-two-levels}
Let $\vec{x},\vec{y},\vec{z}$ be three disjoint vectors of variables. Suppose that $P_i(\vec{x}),Q_i(\vec{y}),R_i(\vec{z})$ are polynomials satisfying
\begin{equation} \label{eq:PQR}
 \sum_{i=1}^t P_i(\vec{x}) Q_i(\vec{y}) R_i(\vec{z}) \equiv 0.
\end{equation}
Let $W_1,\dotsc,W_k\in\mathbb{F}[\vec{x},\vec{y}]$, where $k \leq t$, be a basis for $\{P_i(\vec{x}) Q_i(\vec{y}) \mid i \in [t]\}$. In particular, let $S_i(\vec w)=\sum_{j=1}^k \sigma_{ij} w_j$ be the expansion of $P_i(\vec{x}) Q_i(\vec{y})$ in this basis, that is, $P_i(\vec{x}) Q_i(\vec{y})=S_i(\vec W(\vec{x},\vec{y}))$. Then
\begin{equation} \label{eq:SR}
 \sum_{i=1}^t S_i(\vec{w}) R_i(\vec{z}) \equiv 0.
\end{equation}
\end{lemma}
\begin{proof}
Let $T_j(\vec z)=\sum_{i=1}^t \sigma_{ij}R_i(\vec z)$.
Then 
\[
\sum_{j=1}^k w_jT_j(\vec z) = \sum_{i=1}^t\sum_{j=1}^k \sigma_{ij}w_jR_i(\vec z)=\sum_{i=1}^t S_i(\vec w)R_i(\vec z)
\]
is the polynomial that we are proving to be identically zero.
Assuming the contrary, we conclude that for some $j=j^*$, the polynomial
$T_{j^*}$ is not identically zero.

If $T_{j^*}$ would be multilinear (as it is in our applications, where all $R_i$'s are linear), 
that would already
be enough to reach a contradiction: there should be some vector $\vec\rho$ of 0/1-values 
such that $T_{j^*}(\vec\rho)\neq 0$. Let us substitute $\vec\rho$ for $\vec z$
and $\vec W(\vec x,\vec y)$ for $\vec w$ in (\ref{eq:SR}).
Under this substitution, (\ref{eq:SR}) and (\ref{eq:PQR}) turn into the same equation,
which shows a linear dependency of $W_j$'s
contradicting the fact that $\{w_j \mid j\in[k]\}$ is a basis. 

In order to prove the statement without the multilinearity condition,
choose an extension field that is large enough so that we could choose
a vector $\zeta$ of values in this field such that $T_{j^*}(\zeta)\neq 0$,
and perform the same substitution obtaining a linear dependency
of $W_j$'s over the extension field and hence in $\mathbb{F}$.
\end{proof}

With this lemma at hand, we are ready to prove the simulation
theorem.

\begin{proof}[Proof of Theorem~\ref{th:extpc-sim-hit}]
We consider a hitting formula $H=\{C_1,\dotsc,C_m\}$
and translate each its clause $C_j\in H$
into a product $\prod_{i=1}^n P^i_j(x_i)$,
where 
\[
P^i_j(x_i) := 
\begin{cases}
1,    &\textrm{if $x_i$ does not occur in $C_j$,}\\
x_i,  &\textrm{if $x_i$ occurs in $C_j$ negatively,}\\
1-x_i,&\textrm{if $x_i$ occurs in $C_j$ positively.}
\end{cases}
\]
We call this product a monomial, because in {\extpc}
it can be represented using dual variables $x_i^R=1-x_i$.
We are going to refute $\sum_{j=1}^m \prod_{i=1}^n P^i_j(x_i)$ in {\extpc},
namely, we derive the polynomial 1 from it.

Consider the vector space spanned by the set $\{P^1_j(x_1) P^2_j(x_2) \mid j\in[m]\}$.
We can find a basis $\{Y^{1,2}_i(x_1,x_2) \mid i\in[r]\}$
and introduce extension variables for its polynomials,
$y^{1,2}_i=Y^{1,2}_i(x_1,x_2)$.
We then consider the linear functions $P^{1,2}_j(Y^{1,2}(x_1,x_2))$ giving 
the expansion of $P^1_j(x_1) P^2_j(x_2)$
over this basis, and we can derive in {\extpc} that
$P^{1,2}_j(\vec y^{1,2})-P^1_j(x_1) P^2_j(x_2)=0$.

Recall that 
\begin{equation}\label{eq:extpc-sim-start}
\sum_{j=1}^m \prod_{i=1}^n P^i_j(x_i)
\end{equation}
is identically $1$,
because this is a multilinear polynomial that equals $1$ pointwise on $\{0,1\}^n$
($2^n$ values uniquely define $2^n$ coefficients of the multilinear polynomial).
Then Lemma~\ref{lem:merging-two-levels} shows that
\begin{equation}\label{eq:extpc-sim-onelevel}
\sum_{j=1}^m P^{1,2}_j(\vec y^{1,2})\prod_{i=3}^n P^i_j(x_i)
\end{equation} 
is also identically $1$.

Since the new variables are not multiplied by each other in our monomials,
we can continue this process merging the variables $\vec y^{1,2}$ with $x_3$,
then the new variables $\vec y^{1,2,3}$ with $x_4$, and so on,
until we merge all variables into $\vec y^{[n]}$.
That is, we eventually arrive at 
\begin{equation}\label{eq:extpc-sim-finish}
\sum_{j=1}^m P^{[n]}_j(\vec y^{[n]})
\end{equation}
for a \emph{linear} function $P^{[n]}_j$.
This linear polynomial is also identically $1$.

It is easy to see that {\extpc} proves efficiently that all these polynomials
are equivalent (switching between dual variables and their definitions whenever needed
within two layers of variables), 
in particular, it derives efficiently (\ref{eq:extpc-sim-finish})
from (\ref{eq:extpc-sim-start}).

Since (\ref{eq:extpc-sim-finish}) is a linear polynomial that is identically 1,
it has all zero coefficients except for the free term that is equal to $1$.
\end{proof}

The proof of Theorem~\ref{th:extpc-sim-hit} can be used for
proving in {\extpc} similar statements about multilinear polynomials
that use dual variables. In particular, it can be used for 
simulating {\hitodd} and {\hitk{k}}. 

\begin{proposition}\label{prop:hitodd-verification}
{\hitodd} proofs can be verified in deterministic polynomial time.
{\extpc} over $\GF(2)$ p-simulates {\hitodd}.
In particular, {\EF} p-simulates {\hitodd}.
\end{proposition}
\begin{proof}The proof of Theorem~\ref{th:extpc-sim-hit} works in particular over $\GF(2)$.\end{proof}
This argument allows us to verify \emph{unsatisfiable} odd hitting formulas.
However, a similar technique also makes it possible to check arbitrary formulas
for being odd hitting.
\begin{proposition}\label{prop:hitodd-formula-verification}
Given a formula in CNF, it can be checked in deterministic polynomial time
whether $F$ is an odd hitting formula.
\end{proposition}
\begin{proof}
We need to check that there is no falsifying assignment that falsifies
an even number of clauses. For each clause $C\in F$, substitute
the negation of $C$ as an assignment into $F$, and drop the identically false
clause resulting from substituting $\bar{C}$ into $C$; denote this formula $F_C$. 
Then check that falsifying assignments falsify an even number of clauses of $F_C$
by verifying the identity $\sum_i M_i \equiv 0$ over $\GF(2)$, where $M_i$'s are pseudomonomials (non-negative juntas)
corresponding to the clauses of $F_C$. If for some $C$ the identity is false,
then $F$ had a falsifying assignment that satisfied an even number of clauses.
\end{proof}

\begin{theorem}\label{th:hitk-simulation}
{\hitk{k}} proofs can be verified in deterministic polynomial time.
{\extpc} over a finite field p-simulates {\hitk{k}}.
In particular, {\EF} p-simulates {\hitk{k}}.
\end{theorem}
\begin{proof}
For a hitting-$k$ formula $\bigwedge_{i=1}^m T_i$, by the inclusion-exclusion formula
\[
 \sum_{\emptyset \neq I \in 2^{[m]}} (-1)^{|I|+1} \prod_{i \in I} T_i = 1,
\]
where we abuse the notation by identifying a clause and its {\pcr} translation 
into a monomial that uses dual variables.
Note that the terms containing more than $k$ clauses $T_i$'s are zeros.

Now we can proceed by analogy with the proof of Theorem~\ref{th:extpc-sim-hit}
and Corollary~\ref{cor:ef-sim-hit}.
\end{proof}

%\TODOJ{{\hitk{k}}-{\res},{\hitodd}-{\res}}
%\TODO{In principle, we could speak about {\hitk{k}}-{\res} as well.}
%-----------------------------------------------------------------------------
\subsection{Bonus: succinct proofs and efficient verification of static (semi)algebraic proof systems}\label{sec:bonus}

The proof of Theorem~\ref{th:extpc-sim-hit} does not just provide an {\extpc} proof,
it provides a deterministic polynomial-time verification procedure for
polynomial identity testing for multilinear identities modulo $x+\bar{x}=1$. It can be used in other settings, for example,
for verifying proofs in static (semi)algebraic systems.

Historically, deterministic polynomial-time verification
of proofs in such systems has not been a major concern, because proving a superpolynomial lower bound for such a
system implies a lower bound for a variety of systems that emerge from
supplementing the basic static system with additional means of verification,
for example, the axioms of the polynomial ring as in \cite{GH03}.
However, an increased interest to automated search for sum-of-squares-based proofs
reveals the need for an efficient deterministic formal proof verification procedure. 
A typical static proof constitutes a formal combination
of polynomials including non-negative juntas written in a formal way
(without opening the parentheses). 
To verify such a proof one needs to check 
that this polynomial is identical to a constant.
Opening the parentheses would not work as it would produce
far too many monomials.

Fortunately, the proof of Theorem~\ref{th:extpc-sim-hit} demonstrates 
that any multilinear polynomial identity using dual variables over a finite field
can be verified efficiently. In this subsection we show how to use this idea
for static proof checking.

Verifying {\ns} proofs is easy: it suffices to open parentheses in products
of two polynomials, each of them being represented as a sum of monomials with coefficients.
Alekhnovich et al \cite{ABSRW02} suggested using \emph{dual} variables in {\pcr},
essentially adding extension axioms for the negations of variables to {\pc}.
Such dual variables can be (and have been) also used in other algebraic
and semialgebraic proof systems, in particular, {\ns} turns
into a more powerful system {\nsr} \cite{dRLNS21}.
It is more difficult to verify {\nsr} proofs, however,
the proof of Theorem~\ref{th:extpc-sim-hit} already does it over a finite field:
opening the parentheses, as before, without expanding the definitions of
dual variables turns the proof into a sum of monomials involving the input and dual
variables, exactly as studied in the proof of Theorem~\ref{th:extpc-sim-hit}.
In fact, the identity being verified is a \emph{succinct} {\nsr} proof
of the form \begin{equation}\label{eq:ns}\sum f_ig_i \equiv 1\end{equation}
\emph{without} explicitly mentioning the Boolean axioms $x^2-x=0$ or the dual
variables axioms $x+\bar{x}-1=0$,
that is, $f_i$'s are only the translations of the original clauses.
To verify this identity using the framework of Theorem~\ref{th:extpc-sim-hit},
we turn $f_ig_i$'s into multilinear polynomials by dropping monomials
containing dual variables $x$ and $\bar{x}$ together
and reducing the degrees of variables to one in other monomials.

When we switch to $\mathbb{Q}$ in {\nsr} or turn our attention to
semialgebraic proof systems, there is a problem with this approach:
the coefficients can grow in our transformations between the bases,
and the bit-size of the new proof can become superpolynomial
in terms of the bit-size of the original proof
(note that according to the Cook--Reckhow definition
we are taking into account the bit-size of the proof and
not just the degree or the number of monomials as they sometimes do
in the context of algebraic proofs).
To avoid this obstacle, we can proceed as follows:
\begin{itemize}
\item transform the proof from rationals to integers 
by multiplying it by all the denominators,
the free term $1$ or $-1$ then becomes a different positive/negative constant, 
\item use the Chinese Remainder Theorem to employ verification in finite fields. 
\end{itemize}
Let $M$ be an upper bound on the absolute values of coefficients before monomials in (\ref{eq:ns})
or a similar equation in semialgebraic proof systems.
Then by the Chinese Remainder Theorem, it suffices to verify that $P \equiv 0$ modulo primes whose product exceeds $2M+1$. Using a deterministic polynomial-time primality algorithm, we can find all primes up to $K \log(M) \log \log(M)$ in polynomial time. For an appropriate value of $K$, their product exceeds $2M+1$. Then it is enough to verify the required identity modulo each of these primes, that is, in the respective finite fields.

Common static systems proofs can be verified 
as multilinear polynomial identities using dual variables
by considering succinct proofs and multilinearization as above.
\begin{description}
\item[Static {\ls}, {\sa}, {\scs}.] The static Lov\'asz--Schrijver~\cite{GHP02} and Sherali--Adams~\cite{DMR09,ALN16}, and Subcube Sums proof systems~\cite{FMSV20} are defined as follows.

A \emph{pseudomonomial} is a product of input variables and their negations (in the form of $1-x$). A \emph{conic junta} is a nonnegative linear combination of pseudomonomials. We can generalize proofs using juntas right away by allowing dual variables (and then juntas become simply monomials) and augmenting the system with axioms $x+\bar{x}-1=0$, as is done in \cite{DMR09,dRLNS21}. We can also consider succinct proofs as above by formulating identities modulo $x^2-x=0$ and $x+\bar{x}-1=0$ for every variable $x$.

Static {\ls} is defined for a system of inequalities $s_i\ge0$
that typically include the (obvious) translations of Boolean clauses as linear inequalities
along with inequalities $x^2-x\ge0$, $x-x^2\ge0$, $x\ge0$, $1-x\ge0$ for every variable. 
It is defined as a formal sum-of-products $\sum s_i J_i \equiv -1$,
where $J_i$'s are conic juntas. The dual variables version can be defined
similarly by augmenting the list of $s_i$'s with $x+\bar{x}-1\ge 0$, $1-x-\bar{x}\ge0$.

Sherali--Adams proofs are considered for systems of polynomial equations
(the translations of clauses into pseudomonomials $C_i$ along with $x^2-x=0$
for every variable), the proof is a formal sum-of-products
$\sum C_i (P_i-Q_i) + R \equiv -1$. It can be viewed as a static Lov\'asz--Schrijver
proof, where $C_i=0$ is represented by the two inequalities $C_i\ge0$ and $-C_i\ge0$.
Note that ``additive'' (linear inequalities) 
and ``multiplicative'' (equations for pseudomonomials)
representations of clauses are in fact equivalent with respect to size~\cite[Lemmas 3.1, 3.2]{GHP02},
where efficiently representing long clauses is, of course, only possible with dual variables.

A Subcube Sums proof can be viewed a restriction of succinct unary {\sa} where $P_i-Q_i=1$ and only the size of $R$ counts towards the size of the proof, with copies of the same monomial counted with multiplicity. For instance, a {\scs} proof of a hitting formula is simply ``$0$''. Similarly to {\hit}, this is an inherently succinct proof system.

Also, non-propositional versions of these systems are available, where the axioms are not translations of clauses of a Boolean formula, and multiplying such axioms is allowed, which would a priori cause a problem with our strategy. However, for translations of clauses, the identities that we need to verify are still linear combinations of monomials in the input and dual variables, so they fall completely under the same framework.

\item[Static {\ls}$_+^{\infty}$, {\sos}.] These systems augment the previous systems with squares.
A proof in {\ls}$_+^{\infty}$ is an identity of the sort
\[\sum s_i J_i + \sum_j q_j^2 J'_j\equiv -1\]
where $q_j$ are arbitrary polynomials and $J_i$'s and $J'_j$'s are conic juntas.
There are several definitions of {\sos} in the literature,
but modulo the equivalence between equations and pairs of opposite inequalities,
they are essentially the same in the propositional case.
These systems still fall under our framework, 
as opening the parentheses in a square of a sum of monomials,
$(\sum a_iM_i)^2 = \sum a_i^2M_i^2 + \sum (2a_ia_j)M_iM_j$,
gives only a quadratic number of monomials.
\end{description}
%-----------------------------------------------------------------------------
%\subsection{Non-automatability}
%
%\TODO{We do not have a proof or a short road to it. So it will not appear in this version.}

%=============================================================================
\section{\texorpdfstring{{\hit}}{Hitting} vs \texorpdfstring{{\tres}}{tl-Res} and other classical systems}

In this section we prove that while {\hit} p-simulates {\tres},
in the other direction {\tres} simulates {\hit} only quasi-polynomially.
Moreover, {\tres} is quasi-polynomially weaker than {\hit}.

We also relate {\hit} to other proof systems:
the tree-like version of {\resp} (they are incomparable),
certain versions of {\ns}, and {\res}.

%-----------------------------------------------------------------------------
\subsection{\texorpdfstring{{\tres}}{tl-Res} quasi-polynomially simulates \texorpdfstring{{\hit}}{Hitting}}

We can use a construction of small decision trees from DNF covers of Boolean functions to quasi-polynomially simulate {\hit} in {\tres}.

\begin{theorem}[\cite{EhrenfeuchtHaussler}] \label{th:eh}
  Let $D_0$ and $D_1$ be DNF covers of size $r$ and $s$ that cover the 0s and 1s of a function $f\colon\{0,1\}^n\to \{0,1\}$  respectively. Then there is a decision tree computing $f$ of size at most $2(rsn)^{\log(r+s)+1}$.
\end{theorem}

% \begin{theorem}\label{th:tres-qsim-hit-suboptimal} If a CNF formula $F$ has a {\hit} refutation of size $m$, then $F$ has a {\tres} refutation of size at most $O(2^{3\log^3 m+o(\log^3m)})$.
% \end{theorem}

% \begin{proof}
%   Let $G$ be a hitting formula with $n$ variables and $m$ clauses that is a refutation of $F$. We build a decision tree for the falsified clause search problem of $G$ as follows. Consider the $\log m$ Boolean functions $f_0,\dotsc,f_{\log m-1}$ that output the binary representation of the index of the clause falsified by an assignment. Each function $f_b$ has DNF covers of size $r+s \leq m$, therefore by Theorem~\ref{th:eh} $f_b$ is computable by a decision tree of size at most $(m^2n)^{\log m}$. By composing the decision trees sequentially we obtain a decision tree of size at most $(m^2n)^{\log^2 m}\le 2^{3\log^3 m}$ solving the search problem of $G$. The same decision tree solves the search problem of $F$.
% \end{proof}

%We can also adapt the argument of Ehrenfeucht--Haussler to obtain a more direct proof. An accurate induction also gives slightly better parameters.

We adapt the argument of Ehrenfeucht--Haussler to prove the following theorem.
Intuitively, every hitting formula defines a subcube partition of the Boolean cube $\{0,1\}^n$. The structure of this partition can be used to greedily construct a decision tree ({\tres} refutation) that always queries the most conflicting variable in the narrowest clause. 

\begin{theorem}\label{th:tres-qsim-hit} If a CNF formula $F$ has a {\hit} refutation of size $m$, then $F$ has a {\tres} refutation of size at most $O(2^{2\log^3 m})$.
\end{theorem}

\begin{proof}
  Let $G$ be a hitting formula with $n$ variables and $m$ clauses that is a refutation of $F$ and let us recursively build a decision tree that solve the falsified clause search problem of $G$. As every clause of $G$ is a weakening of some clause of $F$ the resulting tree will also solve the falsified clause search problem of $F$.
  Let $C$ be a narrowest clause in $G$, which has width at most $\log m$ (otherwise a union bound shows that the formula is satisfiable). Since $G$ is a hitting formula, every other clause has at least one contrary literal with respect to $C$. Let $\ell_i$ be the literal in $C$ that appears in the maximal number of other clauses with a different sign, which by an averaging argument is at least $(m-1)/\log m$ times. The decision tree queries the corresponding variable. If the answer does not satisfy $C$ then it satisfies all clauses containing $\olnot{\ell_i}$, hence the resulting formula has $n-1$ variables and at most $m - (m-1)/\log m$ clauses. Otherwise the formula has $n-1$ variables and at most $m-1$ clauses.

  The number of leaves in the decision tree satisfies the recurrence 
  \[
    S(n,m) \leq S(n-1,m - (m-1)/\log m) + S(n-1,m-1).
  \] 
  We claim that $S(n,m) \leq n^{2\log^2 m}$. Indeed,
  \begin{align*}
    S(n,m)
    &\leq S(n-1,m - (m-1)/\log m) + S(n-1,m-1) \\
    &\leq (n-1)^{2(\log(m - (m-1)/\log m))^2} + (n-1)^{2\log^2 (m-1)} \\
    &\leq (n-1)^{2(\log(m(1- 1/(2\log m))))^2} + (n-1)^{2\log^2 m} \\
    &= (n-1)^{2(\log m+\log(1-1/(2\log m)))^2} + (n-1)^{2\log^2 m} \\
       &\leq (n-1)^{2\log m\cdot(\log m + \log(1-1/(2\log m)))} + (n-1)^{2\log^2 m} \\
    &\leq (n-1)^{2\log m\cdot(\log m - 1/(2\log m))} + (n-1)^{2\log^2 m} \\
       &= (n-1)^{2\log^2 m - 1} + (n-1)^{2\log^2 m} \\
    &\leq (n-1+1)^{2\log^2 m}. \qedhere
  \end{align*}
\end{proof}

\begin{remark}
An alternative, more combinatorial way to compute the number of leaves of the tree is as follows. 
At a node of type $(n',m')$ (i.e., with $n'$ variables and $m'$ clauses), we have a left turn to a node of type $(n'',m'')$, where $n'' < n'$ and $m'' \leq (1-1/(2\log m')) m' \leq (1-1/(2\log m)) m'$, and a right turn to a node of type $(n'',m'')$, where $n'' < n'$ and $m'' < m'$. Every path from the root to a leaf (a node of type $(n',1)$) thus contains at most $\log_{(1 - 1 / (2\log m))^{-1}} m = \log m / \log (1 - 1 / (2\log m))^{-1} = O(\log^2 m)$ left turns. We can identify each leaf with a binary string of length at most $n$ with $O(\log^2 m)$ many zeroes. The number of leaves is thus at most
\[
 \binom{\leq n}{\leq O(\log^2 m)} \leq n^{O(\log^2 m)}.
\]
\end{remark}
The argument can be extended to {\hitk{k}} simulation by {\hit} (and hence by {\tres}). 
\begin{proposition}\label{prop:hitk-qsim-hit}
    {\hit} quasi-polynomially simulates {\hitk{k}} up to $k=(\log m)^{O(1)}$.
\end{proposition}
\begin{proof}
Consider a hitting-2 formula with $m$ clauses. If it contains fewer than three clauses then it mentions constantly many variables, and so can be refined to a hitting formula of constant size. Otherwise, it mentions at most $m$ variables, and the narrowest three clauses $T_i,T_j,T_k$ each contain at most $\log_2 m + O(1)$ literals. Consider sets of assignment $S_i, S_j, S_k$ falsified by these clauses. If for some other clauses $T_a$, $T_b$ their respective sets $S_a,S_b$ are not disjoint, then $S_a \cap S_b \cap S_c = \emptyset$ for some $c \in \{i,j,k\}$, and so one of the literals in $T_c$ appears negated in one of $T_a,T_b$. Hence we can find a variable assignment such that for $\Omega(1/\log m)$ fraction of non-disjoint pairs $S_a,S_b$, the variable assignment removes one of the clauses. Construct a decision tree whose root asks about this variable, and recurse. A ``left turn'' is the answer which hits the $\Omega(1/\log m)$ fraction. After $O(\log^2 m)$ left turns, the remaining sets are disjoint (the clauses form a hitting formula), which is a leaf in our decision tree. Since the depth is at most $n$, the decision tree contains $\binom{\leq n}{\leq O(\log^2 m)}$ many leaves, and so the original formula can be refined to a hitting formula of size $n^{O(\log^2 m)}$.

A similar argument works for reducing  hitting-$t$ formulas to  hitting-$(t-1)$ formulas. In the base case, fewer than $t$ terms, the formula mentions fewer than $t-1$ variables, and so there is a refinement of size $2^t$. The decision tree contains $n^{O(t\log^2 m)}$ many leaves, and altogether the reduction loses a multiplicative factor of $n^{O(t\log^2 m)}$. We can reduce all the way to a hitting formula of size $n^{O(t^2\log^2 m)}$.

{By analogy, {{\hitk{{\log^km}}}} can be simulated in {\hit} with complexity $n^{O(\log^{k+2}m)}$.}
\end{proof}
%Finally, consider $p(m)$-hitting formulas. If there are $t$ terms with non-empty intersection then there are at least $2^t$ sets of terms with non-empty intersections. Therefore hitting~$p(m)$ is subsumed by hitting~$\log p(m) = O(\log m)$, and so we can convert to a hitting formula of size $n^{O(\log^4 m)}$.

Later in Theorem~\ref{th:tres-l-b} we prove that the simulation of {\hit} by {\tres} cannot be polynomial; however, we do not know whether it can be improved to $m^{O(\log m)}$.

\begin{corollary}\label{cor:hit-vs-res}
There are formulas that have polynomial-size {\res} proofs but require exponential-size {\hit} proofs.
\end{corollary}
\begin{proof}
\cite{BIW04} proves that there are formulas that have polynomial-size {\res} proofs but require exponential-size {\tres} proofs. Using Theorem~\ref{th:tres-qsim-hit} the statement follows.
\end{proof}
\begin{remark}
Similarly to Theorem~\ref{th:tres-qsim-hit},
{\hitres} can be quasi-polynomially simulated in {\res}
(every hitting resolution step can be simulated using  Theorem~\ref{th:tres-qsim-hit}),
%\TODO{should we expand on it?}
and thus an exponential-size lower bound for it also follows from
exponential-size lower bounds for {\res} (e.g., \cite{Urq87}).
\end{remark}
%\TODO{Do we want to say anything about Merge/Split and MaxRes?} \RED{YF: Probably not.}

%-----------------------------------------------------------------------------
\subsection{{\hit} is quasi-polynomially stronger than {\tres}}
\begin{theorem}\label{th:hit-sim-tres} {\hit} p-simulates {\tres}.
\end{theorem}
\begin{proof}
A {\tres} proof can be viewed as a decision tree \cite{BIW04}:
every application of the resolution of clauses $C\vee x$ and $D\vee\bar{x}$ by the variable $x$
is considered as a decision by the variable $x$, so that the decision
$x=0$ leads to the node $C\vee x$ from the node $C\vee D$,
and the decision $x=1$ leads to the node $D\vee\bar{x}$.
Therefore, the assignment of decisions starting at the root, labelled by the final empty clause,
and ending at a leaf, labelled by a clause $L$ of the input formula $F$,
falsifies $L$. The negation of such assignment, viewed as a clause $N$,
is a weakening of $L$. The conjunction of all such $N$'s 
is an unsatisfiable hitting formula $H$ that is a proof of $F$.
The number of clauses in $H$ is at most the number of leaves in the decision tree and therefore at most
the number of occurrences of $F$'s clauses in the {\tres} refutation.
\end{proof}

We use $\oplus$-lifting to prove our separation result.
We need three folklore or easy statements that we include for the sake of completeness.

\begin{definition}[Composition of multivariate functions] For Boolean functions $f\colon \{0,1\}^n \to \{0,1\}$ and $g\colon\{0,1\}^m \to \{0,1\}$, we define $f \circ g^n\colon\{0,1\}^{nm} \to \{0,1\}$ as \[(f \circ g^n)(x^1_1, \dots, x^1_m, \dots, x^n_1, \dots, x^n_m) := f(g(x^1_1, \dots, x^1_m), \dots, g(x^n_1, \dots, x^n_m)).\]
\end{definition}

%\TODOJ{unsure it's folklore, probably we need a reference}
\begin{lemma}\label{lem:xor_dt}
If $f\colon \{0,1\}^n \to \{0,1\}$ is a function requiring decision tree depth $d$, then $f \circ (\oplus_2)^n$ requires decision tree size $2^{d}$.  
\end{lemma}
\begin{proof}
We consider a size-$s$ decision tree $T$ computing $(f \circ (\oplus_2)^n)(x^1_1, x^1_2, \dots, x^n_1, x^n_2)$
and transform it into a querying strategy for $f$ using at most $\log_2 s$ queries.
Starting at the root, we go down the tree
both querying $f(\vec y)$'s inputs $y_i$ and giving answers
to the queries made in the original tree $T$.
We will be interested in the number of leaves in the current subtree of $T$.

Assume that $T$ queries $x^i_j$. If the other variable $x^i_{2-j}$ has not been queried yet,
we do not query inputs of $f$ and simply choose the answer $\tau\in\{0,1\}$ for $x^i_j = \tau$
that directs us to the subtree that has fewer leaves.
Otherwise, we choose $\tau = y_i\oplus x^i_{2-j}$.
Therefore, we go into the subtree of $T$ containing at least twice fewer leaves
at least once per query to $y_i$'s, and we can do it at most $\log_2 s$ times
before we come to a leaf.
\end{proof}

\begin{lemma}\label{lem:unambiguous-certificate-composition}
  Suppose a function $f\colon\{0,1\}^n \to \{0,1\}$ has an unambiguous DNF of width $w$ and $g\colon\{0,1\}^m \to \{0,1\}$ has a decision tree of depth $d$. Then $f \circ g^n$ has an unambiguous DNF of width $w d$.
\end{lemma}
\begin{proof}
A decision tree of depth $d$ can be represented as an unambiguous DNF of width $d$ itself,
so we have a set $G$ of terms (conjunctions) of width at most $d$ contradicting each other,
and it is split into two sets $G_0$ and $G_1$ equivalent to $g(\vec{x})=0$ and $g(\vec{x})=1$, respectively.

Consider a term in the unambiguous DNF representation of $f(\vec y)$ and replace each literal $y_i = \alpha$ with unambiguous DNF representation of $g(\vec{x}^i) = \alpha$ obtained from renaming the variables in $G_\alpha$, and then expand.  
\end{proof}

\begin{lemma}\label{lem:unambiguous-width2size}
Suppose an unambiguous DNF has width $w$. Then it has at most $2^w$ terms.
\end{lemma}
\begin{proof}
  Let $n$ be the number of variables in the DNF. Then each term has at least $2^{n-w}$ satisfying assignments. Then since no assignment satisfies two terms, the number of terms is at most $\frac{2^n}{2^{n-w}} = 2^w$.
\end{proof}

In order to prove our result, we need the following 
separation of randomized query complexity (deterministic is enough for our purpose)
and unambiguous certificate complexity from \cite{AKK16}.

\begin{definition}[{\cite{ABK16,AKK16}}]
Let $f\colon\{0,1\}^N\to\{0,1\}$ be a function,
$c = 10 \log N$ and $m = c\cdot \cert(f) \log N = 10\cert(f)\log^2 N$.
Then the cheat sheet version of $f$, denoted
$\fcs$, is a total function
$\fcs\colon(\{0, 1\}^N )^ c \times (\{0, 1\}^m)^{2^c} \to\{0,1\}$.

Let the input be written as $(x^1,x^2,\dotsc,x^c,Y_1,Y_2,\dotsc,Y_{2^c})$,
where for all $i\in[c]$, $x^i\in\{0,1\}^N$, and for all $j\in[2^c]$,
$Y_j\in\{0,1\}^m$.
Let $\ell_i = f(x^i)$ and $\ell\in[2^c]$
be the positive integer corresponding to the binary string 
$\ell_1,\dotsc,\ell_c$.
Then we define the value of
$\fcs$
to be 1 if and only if $Y_\ell$ contains certificates for 
$f(x^i)=\ell_i$ 
for all $i\in[c]$.
\end{definition}

At first glance, the definition of $\fcs$ might look nonconstructive
due to the usage of $\cert(f)$. However, the theorem of \cite{AKK16} uses
an appropriate upper bound on $\cert(f)$, which is proved along
with the interactive construction of the function.

\begin{theorem}[{\cite[Theorem 5.1]{AKK16}}]\label{th:AKK16} 
Let $f_0 = \AND_n$ and $f_k$ be defined inductively as $f_k := \AND_n \circ (\OR_n \circ f_{k-1})_{\mathrm{CS}}$, where $f_k$ has $O(n^{25^k})$ inputs.
Then $\R(f_k) = \tilde\Omega(n^{2k+1})$ and $\ucert(f_k) = \tilde O(n^{k+1})$.
%Hence there is a function f with R(f) \ge \ucert(f)^{ 2−o(1)}.
\end{theorem}

\begin{theorem}\label{th:tres-l-b} 
For every $\varepsilon>0$,
there exists a sequence of unsatifiable hitting formulas $G_m$ 
containing $2^{\tilde O(m)}$ clauses of width at most $\tilde O(m)$
such that 
$G_m$ requires {\tres} proof size $2^{\tilde\Omega(m^{2-\varepsilon})}$.
\end{theorem}

\begin{proof}
We consider the composition of a function separating $\R$ and $\ucert$ with the parity function.
Let $f$ be the function given by Theorem~\ref{th:AKK16}. Let $F_0$ and $F_1$ be the respective unambiguous DNFs, they both have width $\tilde{O}(n^{k+1})$ and thus by Lemma~\ref{lem:unambiguous-width2size} have size $2^{\tilde{O}(n^{k+1})}$. Therefore, 
Lemma~\ref{lem:unambiguous-certificate-composition} provides two unambiguous DNFs, $F_0 \circ (\oplus_2)^m$ representing $f \circ (\oplus_2)^m$ and $F_1 \circ (\oplus_2)^m$ representing $\bar{f} \circ (\oplus_2)^m$, both of width $\tilde{O}(n^{k+1})$ and of size $2^{\tilde{O}(n^{k+1})}$.

\vspace{2mm}
Consider a CNF $F = \overline{F_0 \circ (\oplus_2)^m \lor F_1 \circ (\oplus_2)^m}$.
This is an unsatisfiable hitting formula, so it is a short {\hit} proof of itself.
It contains $2^{\tilde{O}(n^{k+1})}$ clauses.

\vspace{2mm}
Suppose $F$ has a {\tres} refutation of size $s$. Let us view it as a decision tree solving the falsified clause search problem for $F$. Now let us change leaf labels in the following way: a leaf labeled with a clause that came from $F_i \circ (\oplus_2)^m$ gets label $i$. It is easy to see that the resulting tree computes $f \circ (\oplus_2)^m$. Recall that $f$ is constructed by Theorem~\ref{th:AKK16} and its query complexity is $\tilde{\Omega}(n^{2k+1})$. By Lemma~\ref{lem:xor_dt} its decision tree must have size $2^{\tilde{\Omega}(n^{2k+1})}$. Now the theorem claim holds for $m=n^{k+1}$.
\end{proof}

\subsection{{\hit} and {\tres}(\texorpdfstring{$\oplus$}{⊕}) are incomparable}
\label{sec:linear-hitting-vs-tree-like-resolution}

\subsubsection{A hard formula for {\hit}}
\label{sec:hitting-xor-vs-resolution}

In this section we observe that there exist formulas that are easy for {\tresp} and exponentially hard for {\hit}. For this, we recall the separation of {\tresp} from {\res} shown in \cite{IS20} for Tseitin formulas.

\begin{definition}[Tseitin formulas $T_{G,c}$]\label{def:tseitin}
For a constant-degree graph $G=(V,E)$ and a 0/1 vector $c$ of ``charges'' for the vertices, 
consider the following linear system in the variables $x_{e}$ for $e\in E$:
\[
\bigwedge_{v\in V}\; 
\left(
\bigoplus_{e\ni v} 
x_{e}=c_v 
\right),
\]
where $\bigoplus_{v\in V} c_v=1$.
In the corresponding Tseitin formula $T_{G,c}$ in CNF
each vertex constraint $\bigoplus_{e\ni v}x_e=c_v$
expands into $2^{\deg v - 1}$ clauses of width $\deg v$.
\end{definition}

\begin{theorem}[\cite{Urq87}]
\label{th:res-lowerbound-tseitin}
    There exists a family of constant-degree graphs $G_n$ with $n$ nodes and a family of charge vectors $c_n$ such that $\mathrm{Ts}_{G_n,c_n}$ requires {\res} refutation of size $2^{\Omega(n)}$.
\end{theorem}

\begin{theorem}[\cite{IS20}]
    For any graph $G$ and charges $c$ the Tseitin formula $\mathrm{Ts}_{G,c}$ has a tree-like Res($\oplus$) refutation of size linear in the size of the CNF.
\end{theorem}

Given the quasi-polynomial simulation of Theorem~\ref{th:tres-qsim-hit}
and the following generalization of Theorem~\ref{th:hit-sim-tres}, 
we can separate {\hit} from {\tresp} and {\hitp}.

\begin{proposition}\label{prop:hitp-sim-tresp}
If $F$ has a tree-like {\resp} refutation of size $s$, 
then it has a {\hitp} refutation of size $s$.
\end{proposition}
\begin{proof}
By analogy with Theorem~\ref{th:hit-sim-tres}, 
the leaves of a tree-like {\resp} refutation form a {\hitp} refutation. 
\end{proof}

\begin{corollary}
\label{cor:tlresxor-vs-hitting}
  There exists a family of CNF formulas $F_n$ such that $F_n$ requires Resolution refutation of size $2^{\Omega(n)}$, Hitting refutation of size $2^{n^{\Omega(1)}}$ and admits polynomial-size tree-like Res($\oplus$) refutation (and, consequently, polynomial-size {\hitp} refutation).
\end{corollary}
\begin{proof}
  Take $\mathrm{Ts}_{G_n, c_n}$ from Theorem~\ref{th:res-lowerbound-tseitin} and apply Prop.~\ref{prop:hitp-sim-tresp}.
\end{proof}

\subsubsection{A hard formula for \texorpdfstring{{\tresp}}{tl-Res(xor)}}

In addition to separating {\hit} from {\tres}, we can follow the same plan to separate it from a stronger {\tresp} proof system, that is, to lift a separation between unambiguous certificate complexity and query complexity. We cannot use decision tree size to bound {\tresp} size, but rather the stronger randomized communication complexity measure.

\begin{theorem}[{\cite[Theorem~3.11]{IS20}}]
\label{th:res-parity-communication-sim}
  Let $F$ be an unsatisfiable CNF that has tree-like Res($\oplus$) refutation of size $t$ then the randomized communication complexity of the falsified clause search problem for $F$ is $O(\log t)$.
\end{theorem}

An analogue of Lemma~\ref{lem:xor_dt} for randomized communication complexity also holds, with the difference that we need to compose the DNFs $F_0$ and $F_1$ from Theorem~\ref{th:tres-l-b} with the indexing function instead of $\oplus$.
The indexing function $\indexing_m\colon [m] \times \{0,1\}^m \to \{0,1\}$ is defined as  $\indexing_m(i, x) = x_i$, i.e. it accepts an index and a vector and returns the element of the vector with the given index. Observe that $\indexing_m$ has a decision tree of depth $\lceil \log_2 m \rceil + 1$: we first query the index and then query a single bit of the vector.

\begin{theorem}[\cite{GPW17}]
\label{th:randomized-lifting}
    If a function $f\colon \{0,1\}^n \to \{0,1\}$ requires a randomized decision tree of depth $t$, then the function $f \circ (\indexing_m)^n$ where $m=n^{256}$ requires randomized communication cost $\Omega(t \log n)$.
\end{theorem}

This is all we need to prove the separation.

\begin{theorem}
  \label{th:hitting-vs-treeresxor}
  For every $\epsilon>0$, there is a sequence of unsatisfiable hitting formulas $G_m$ containing $2^{\tilde O(m)}$ clauses of width $\tilde O(m)$ that requires {\tres}$(\oplus)$ proof size $2^{\tilde\Omega(m^{2-\epsilon})}$.
\end{theorem}
\begin{proof}
  We construct $F$ as follows: take $f_k\colon \{0,1\}^M \to \{0,1\}$ from Theorem~\ref{th:AKK16}, where $M=O(n^{25^k})$ and $m=\ucert(f_k)=\tilde O(n^{k+1})$, and compose it with $\indexing_{M^{256}}$. By Lemma~\ref{lem:unambiguous-certificate-composition}, $f_k \circ (\indexing_{M^{256}})^M$ and $\lnot (f_k \circ (\indexing_{M^{256}})^M)$ both have unambiguous DNF representations of width $\tilde{O}(n^{k+1})$. Then $F$ is the conjunction of the negated terms in these representations.

  On the one hand, by Lemma~\ref{lem:unambiguous-width2size} $F$ has size $2^{\tilde{O}(m)}$, and since it is already a hitting formula, it serves as a {\hit} proof of itself.

  On the other hand, let $s$ be the size of the smallest {\tresp} refutation of $F$. Then by Theorem~\ref{th:res-parity-communication-sim} there exists a randomized communication protocol solving the falsified clause search problem for $F$ of cost $O(\log s)$. Observe that such protocol can be easily converted into a protocol solving $f \circ (\indexing_{k^{256}})^k$ with the same cost and at most the same probability of error: if the protocol returns a clause corresponding to a term in the representation of $f \circ (\indexing_{k^{256}})^k$ answer 1, otherwise answer $0$. Then by Theorem~\ref{th:randomized-lifting} we have that $\log s = \tilde{\Omega}(n^{2k+1})$, i.e.
  \[s = 2^{\tilde{\Omega}(n^{2k+1})} = 2^{\tilde{\Omega}(m^{2-\epsilon})}. \qedhere\]
\end{proof}

%-----------------------------------------------------------------------------
\subsection{Relation to {\res} and {\ns}}
\label{sec:hitting-vs-res-ns}

As we discussed in Section~\ref{sec:hitting-xor-vs-resolution}, a corollary of Theorem~\ref{th:tres-qsim-hit}, which shows that {\tres} quasi-polynomially simulates {\hit}, is that if a proof system $\mathcal{P}$ is exponentially separated from {\tres} then $\mathcal{P}$ is also exponentially separated from {\hit}. Since this is the case with {\res} and {\ns}---which have short proofs of the ordering principle and the bijective pigeonhole principle \cite{BR96} respectively, while {\tres} requires exponentially long proofs of both---we conclude that {\res} and {\ns} are exponentially separated from {\hit}.

In this section we explore whether a simulation or separation in the other direction exists. We show that the formula that we used for the quasi-polynomial
separation of {\hit} from {\tres} has short {\res} refutations, and therefore cannot be used
for showing a separation from {\res}.
%\TODO{do we want also to share suggest strategies for proving such a separation result? looks like not in this version, at least.}
We also show that in a sense {\ns} simulates {\hit}.

\subsubsection{Dag-like query complexity of functions}\label{sec:dag-like-query} %required for \res upper bound as well

%\TODO{Ask Dima if this is the appropriate citation.}

Our resolution upper bound is in fact an upper bound on a width, a measure that can be studied through its characterization as a two-player game~\cite{Pud00,AD08}.
Building on such game, Göös et al.~\cite{GGKS20} introduced the following generalization of the query complexity of a function $f\colon \{0,1\}^n \to M$. We view it as a game between two players, the \emph{querier} and the \emph{adversary}. The querier maintains a partial assignment $\rho\colon [n] \to \{0,1,*\}$. At each step the querier can either query a variable $i \in \rho^{-1}(*)$, in which case the adversary picks a value $\alpha \in \{0,1\}$ and assigns $\rho(i) := \alpha$, or forget a variable $i \in \rho^{-1}(\{0,1\})$, assigning $\rho(i) := *$. Note that the adversary may answer differently the next time a forgotten variable is queried. The game ends only if $\rho$ is a certificate for $f$, that is if there exists $m \in M$ such that for all $x \in \{0,1\}^n$ consistent with $\rho$, we have $f(x) = m$. The width of a particular game $\pi$ is $\dagw(\pi)$, the maximal size (number of assigned variables) of $\rho$ at any step of the game, and the leaf-width is $\leafw(\pi)$, the maximum size of $\rho$ at any terminal step of the game.
The dag-like query complexity of $f$ is the width of the game assuming optimal play where querier aims to minimize the width and adversary aims to maximize it. Denote the dag-like query complexity of $f$ by $\dagw(f)$.

This notion can be similarly defined for relations, and for an unsatisfiable CNF $\phi = \bigwedge_{i\in [m]} C_i$ the dag-like query complexity of the falsified clause search relation $\{(x, i) \mid C_i(x) = 0\}$ is exactly the resolution width of $\phi$~\cite{AD08}. If $\phi$ is a hitting formula, this relation is actually a function.

As a preliminary step towards our upper bound we compute the dag-like query complexity of the $\OR \circ \AND$ function, which is a key part of the construction of Theorem~\ref{th:AKK16}. We begin by proving that
dag-like query complexity is sub-multiplicative with respect to composition, and in fact we can be a bit more precise.

\begin{lemma}
\label{lem:composition-and-width}
  Let $f\colon \{0,1\}^n \to \{0,1\}$ and $g\colon \{0,1\}^m \to \{0,1\}$ be two Boolean functions. Let $\pi$ be the dag-like query complexity game for $g$.
  Then \[ \dagw(f \circ g^n) \le (\dagw(f)-1) \cdot \leafw(\pi) + \dagw(\pi) \leq \dagw(f)\cdot\dagw(g). \]
\end{lemma}

\begin{proof}
  We simulate the querier's strategy for $f$ in the following way: suppose at some point the querier for $f$ had an assignment $\rho\colon [n] \to \{0,1,*\}$. Then the querier for $f \circ g^n$ has assignments $\rho'_1, \dots, \rho'_n\colon [m] \to \{0,1, *\}$ such that whenever $\rho(i) = \alpha \neq *$, $\rho'_i$ corresponds to an $\alpha$-certificate of $g$. If the querier for $f$ forgets a variable $i$, we assign $\rho'_i := *^m$, if they query a variable $i$, we run the querying strategy for $g$ within $\rho'_i$ and eventually end up with a certificate for $g$ as an assignment $\rho'_i$, by definition $|(\rho'_i)^{-1}(\{0,1\})| \le \dagw_{\text{out}}$. The maximum number of assigned variables when simulating a query to variable $i$ is then $\dagw(\pi)$ corresponding to assignment $\rho'_i$, and $(\dagw(f)-1) \cdot \dagw_{\text{out}}(\pi)$ corresponding to the remaining assignments.
   \end{proof}

\begin{corollary}
\label{cor:composition-and-width}
  Let $f,g$ be as in Lemma~\ref{lem:composition-and-width}. Then
  \[ \dagw(f \circ g^n) \le (\dagw(f) - 1) \cdot \cert(g) + \min\{\cert_0(g)\cert_1(g), m\}.\]
\end{corollary}
\begin{proof}
  Consider the standard strategy used to bound the query complexity of $g$ in terms of its certificate complexity, which can be described as follows. Let us pick an arbitrary $0$-certificate $\rho$ for $g$ and query all its variables. If the values match $\rho$, we return it, otherwise, we claim that the current partial assignment $\tau$ contains a variable from every $1$-certificate of $g$. Thus $g|_\tau$ has 1-certificate complexity at most $\cert_1(g) - 1$. Let us apply this procedure (query some $0$-certificate that agrees with the current partial assignment) until we find a $0$-certificate or the certificate complexity shrinks to zero. In the latter case, there exists a $1$-certificate which is a restriction of the current assignment.

  Using this strategy for $g$ as the game $\pi$ in Lemma~\ref{lem:composition-and-width} (and forgetting all the variables outside the found certificate just before the game ends) we get $\dagw(\pi) \le \cert_0(g) \cert_1(g)$ and $\leafw(\pi) \le \max\{\cert_0(g), \cert_1(g)\}$. Since we never have to query a variable twice we can assume $\dagw(\pi) \le m$.
\end{proof}

\begin{corollary}
  \label{cor:width-or}
    For $g$ as in Lemma~\ref{lem:composition-and-width} we have
    \begin{align*} 
    \dagw(\AND_n \circ g^n) \le (n-1)\cert_1(g) + \dagw(g), \\
    \dagw(\OR_n \circ g^n) \le (n-1)\cert_0(g) + \dagw(g).
      \end{align*}
\end{corollary}
\begin{proof}
    Consider a decision tree for $\AND_n$ which queries input bits one by one and returns $0$ as soon as it encounters a $0$-bit. Since in this tree all the assignments have at most a single $0$, the cumulative size of the certificates for $g$ that we store until the end of the game is at most $(n-1) \cert_1(g)$. This combined with the argument in Lemma~\ref{lem:composition-and-width} yields the upper bound. The upper bound for $\OR_n$ is analogous.
\end{proof}

\begin{corollary}
  \label{cor:width-or-and}
  $\dagw(\OR_n \circ \AND_n) \leq 2n-1$.
\end{corollary}

We also need to introduce a notion of unambiguous dag-like query complexity, and its one-sided variants. We define $\uw(f)$ to be the dag-like query complexity when querier is limited to strategies where the certificates at terminal states are unambiguous. We define $\uw_0(f)$ (resp. $\uw_1(f)$) to be $\uw(f)$ when the adversary always answers consistently with a $0$-input (resp. a $1$-input).

\begin{lemma}
  \label{lem:uw-and}
  Let $f\colon \set{0,1}^n\to\set{0,1}$ be a Boolean function. Then
  \begin{align*}
  \uw_0(\AND_n \circ g^n) &\leq \uw_0(g) + (n-1)\uw_1(g), \\
  \uw_1(\AND_n \circ g^n) &\leq n \cdot \uw_1(g).
  \end{align*}
\end{lemma}

\begin{proof}
  We use the same composition strategy from Lemma~\ref{lem:composition-and-width}, taking the
  same strategy for $\AND_n$ as in Corollary~\ref{cor:width-or-and}, and an unambiguous game $\pi$ for $g$. At the time of evaluating $g_i$ we have that assignments $\rho'_1,\dotsc,\rho'_{i-1}$ are all unambiguous $1$-certificates, hence assigning up to $(i-1)\uw_1(g)$ many variables, and the bounds follow. The terminal certificates are unambiguous because they are the composition of unambigous certificates. 
  Suppose that it is not the case, so there is some input $x$ that is covered with two certificates $c_1, c_2$.
  Certificates $c_1, c_2$ are the compositions of some unambiguous $\AND_n$ certificates $a_1, a_2$, respectively, 
  with the unambiguous certificates for $g$.
  If $a_1 \neq a_2$ then the composed certificates cannot agree. If $a_1 = a_2$ then $c_1$ and $c_2$ differ in the part of certificate that correspond to some copy of $g$. But that is impossible as certificates for $g$ are unambiguous.
\end{proof}

\subsubsection{Upper bound in {\res}}

To construct a {\res} refutation we first reprove the upper bound part of Theorem~\ref{th:AKK16}---separating unambiguous certificate complexity from randomized query complexity---strengthening it to an upper bound for unambiguous dag-like query complexity in place of unambiguous certificate complexity. We need to make a few minor changes arising from the fact that $\dagw(\AND_n)=n$ while $\cert_0(\AND_n)=1$, but using the fact that $\dagw(\OR_n \circ \AND_n)=O(n)$ and not $\Theta(n^2)$ is enough for our purposes.

\begin{lemma}[\cite{AKK16}]
  \label{lem:unambiguous-cheatsheet}
  The following are unambiguous certificates for $\fcs$.
  \begin{itemize}
  \item $0$-certificates: $c$ unambiguous certificates that $f^c=\ell$ together with the contents of $Y_\ell$.
  \item $1$-certificates: the contents of $Y_\ell$ together with the positions described in $Y_\ell$.
  \end{itemize}
\end{lemma}

\begin{lemma}
  \label{lem:width-cheatsheet}
  Let $f\colon \set{0,1}^N\to\set{0,1}$ be a Boolean function. Then
  \begin{align}
  \dagw(\fcs) &= O(\dagw(f) \log^2 n), \label{eq:dagw-cs}\\
  \uw_0(\fcs) &= O(\uw(f) \log^2 n), \label{eq:uw0-cs}\\
  \uw_1(\fcs) &= O(\dagw(f) \log^2 n). \label{eq:uw1-cs}
  \end{align}

\end{lemma}

\begin{proof}
  To prove~\eqref{eq:dagw-cs} we use the following strategy.
  We query the $c$ copies of $f$ in parallel to obtain a pointer $\ell\in[2^c]$, using width $O(\dagw(f)\cdot c)$. We then query all of $Y_\ell$ and the positions described by $Y_\ell$, checking whether they are indeed a set of certificates for $f^c$, for a total width of $O(\dagw(f)\cdot c+m)$.
  Since we choose $c=10\log N$ and $m=c\cdot \cert(f)\cdot\log N=O(\dagw(f)\log^2N)$, the total width is $O(\dagw(f)\cdot c+m)=O(\dagw(f)\log^2 n)$.

  To get~\eqref{eq:uw1-cs} we follow the same strategy as for~\eqref{eq:dagw-cs}.
  To prove~\eqref{eq:uw0-cs} we use unambiguous strategies to query the $c$ copies of $f$, which increases the width to $O(\uw(f)\cdot c+m)$. By Lemma~\ref{lem:unambiguous-cheatsheet} the certificates we obtain in both cases are unambiguous.
\end{proof}

\begin{lemma}
  \label{lem:dagAKK16}
  Let $k$ be a constant, and let $f_k$ be the function defined in Theorem~\ref{th:AKK16}. Then $\uw(f_k)=\tilde O(n^{k+1})$.
\end{lemma}

\begin{proof}
  Let $g_1 = \OR_n \circ \AND_n$ and $g_k = g_1 \circ (g_{k-1})_\mathrm{CS}$, so that $f_k = \AND_n \circ (g_k)_\mathrm{CS}$.

  Let us first show by induction that $\dagw(g_k)=\tilde O(n^k)$. We already proved the base case $\dagw(g_1)=O(n)$ in Corollary~\ref{cor:width-or-and}.
  Assuming that $\dagw(g_{k})=\tilde O(n^{k})$, the dag-like query complexity of its cheatsheet version is $\dagw((g_{k})_{\mathrm{CS}})=\tilde O(n^{k})$ by Lemma~\ref{lem:width-cheatsheet}.
  Together with the fact that decision-DAG-width is sub-multiplicative with respect to composition proven in~Lemma~\ref{lem:composition-and-width}, this implies that $\dagw(g_{k+1})=\tilde O(n^{k+1})$.

  Now we show by induction that $\uw(f_k)=\tilde O(n^{k+1})$. We already proved the base case $\uw(f_0)=O(n)$ in Lemma~\ref{lem:uw-and}.

For the induction step we have
\begin{align*}
  \uw_0(f_{k+1}) &= \uw_0(\AND \circ (g_{k+1})_\mathrm{CS})
                   \leq \uw_0((g_{k+1})_\mathrm{CS}) + n\cdot \uw_1((g_{k+1})_\mathrm{CS}) \\
  &= \tilde O(\uw(g_{k+1}) + n\cdot \dagw(g_{k+1})) = \tilde O(n\cdot \uw(f_k)+n\cdot \dagw(g_{k+1})) = \tilde O(n^{k+2}+n^{k+2}), \\
  \uw_1(f_{k+1}) &= \uw_1(\AND \circ (g_{k+1})_\mathrm{CS})
                   \leq n\cdot\uw_1((g_{k+1})_\mathrm{CS}) = \tilde O(n^{k+2})
\end{align*}
concluding the proof.
\end{proof}

Finally, we can build the {\res} refutation.

\begin{theorem}\label{th:res-u-b}
The formula $G_m$ of Theorem~\ref{th:tres-l-b} has a {\res} refutation of size $2^{\tilde O(m)}$.
\end{theorem}

\begin{proof}

Let $F_0$ and $F_1$ be the unambiguous DNFs that we obtain from the leaves of the strategy for $f_k$ of Lemma~\ref{lem:dagAKK16}. Since we used the same certificates as in Theorem~\ref{th:AKK16}, these are the same as in Theorem~\ref{th:tres-l-b}, and we only need to modify our strategy to output the certificates themselves rather than $0$ and $1$ in order to obtain a strategy for $S$, the falsified clause search problem of $\olnot{F_0 \lor F_1}$, of query complexity $\dagw(S)=\tilde O(n^{k+1})$. By sub-multiplicativity of query complexity we have that the dag-like query complexity of $S'$, the falsified clause search problem of $G_m$, is also $\dagw(S')=\tilde O(n^{k+1})=\tilde O(m)$. From the equivalence between dag-like query complexity and width we have that $G_m$ has a {\res} refutation of width $\tilde O(m)$, which implies a size upper bound of $\lvert\vars(G_m)\rvert^{\tilde O(m)}=2^{\tilde O(m)}$.
\end{proof}

\subsubsection{Upper bound in {\ns}}

Given a clause
$C=\bigvee_{i\in P}x_i \lor \bigvee_{i\in N}\olnot{x_i}$, let
$p_C = \prod_{i\in P} (1-x_i) \cdot \prod_{i\in N} x_i$ be the
polynomial whose roots are the satisfying assignments of $C$. Recall
that a {\ns} certificate that a set of polynomials $\set{p_i}$ has no common root
is a set of polynomials $\set{q_i}$ such that $\sum p_iq_i\equiv1$, 
and the degree of a certificate is $\max_i \deg(p_iq_i)$. A
{\ns} refutation of a CNF $F$ is a {\ns} 
certificate for $\set{p_C \mid C\in F} \union \set{x_i^2-x_i}$.

It turns out that {\ns} simulates {\hit} with respect to degree.
\begin{proposition} \label{prop:ns-degree-simulation}
  {\ns} degree is at most {\hit} width. 
\end{proposition}

\begin{proof}
  Let $F$ be a CNF and $H$ be a {\hit} refutation of $F$.
  Observe that the polynomial $\sum_{C\in F} p_C$ is identical to $1$
  because $H$ is a partition of the hypercube.
  Fix for each clause $C\in H$ a clause $C' \in F$ such that
  $C' \subseteq C$. Let
  $q_D = \sum_{\{C\in H \mid C'=D\}}p_{C\setminus C'}$.
  We claim that the set of polynomials $\{q_D \mid D \in F\}$ is a
  {\ns} certificate for $F$, and indeed we have
  $\sum_{D\in F}p_Dq_D=\sum_{C\in F}p_C=1$. Furthermore, the degree of
  $p_Dq_D$ is bounded by the degree of $p_C$, which equals the width
  of $H$.
\end{proof}

When measuring the size of a {\ns} refutation it is more appropriate to consider a definition that allows us to introduce dual variables $\bar{x}=1-x$ \cite{ABSRW02} resulting in a new system {\nsr} \cite{dRLNS21}, since otherwise a formula containing a wide clause with many positive literals would already require exponential size when translated to polynomials. We discuss this system in Sect.~\ref{sec:verification}. Moreover, we discuss \emph{succinct} {\nsr} proofs that contain only side polynomials for the input axioms and not for $x^2-x=0$ or $x+\bar{x}-1=0$.
In fact, Prop.~\ref{prop:ns-degree-simulation} already shows that succinct {\nsr} polynomially simulates {\hit} with respect to size.

\section{\hitodd}\label{sec:odd-hitting}
As mentioned in Sect.~\ref{sec:bonus}, {\hitodd} proofs can be verified
similarly to {\nsr} proofs over $\GF(2)$. The two proof systems are similar:
an {\nsr} proof is a Nullstellensatz proof from pseudomonomial equations $m_i=0$
that are translations of the input clauses $C_i$, the Boolean equations
$x_j^2-x_j=0$ for every variable $x_j$, and the axioms $\bar x_j + x_j - 1 = 0$
defining the dual variables. The side polynomials over $\GF(2)$
are just sums of some monomials $q_{ik}$, $r_{j\ell} $, and $s_{jt}$ such that
\[
\sum_i m_i\sum_k q_{ik} +
\sum_j (x_j^2-x_j)\sum_\ell r_{j\ell} +
\sum_j (\bar x_j + x_j - 1)\sum_t s_{jt} 
\equiv 1.
\]
On the other hand, an {\hitodd} proof also can be written using polynomials
over $\GF(2)$ as a sum of
pseudomonomial equations $m_i=0$ multiplied by sums of monomials
(every such product expresses a weakening of $C_i$)
\[\sum_i m_i\sum_k q_{ik} \equiv 1 \pmod {\langle x_j^2-x_j, \bar x_j + x_j - 1\rangle_j}.\]
The difference is that in {\hitodd} the equivalence is only modulo the ideal,
thus {\hitodd} gives \emph{succinct} {\nsr} proofs, as in Sect.~\ref{sec:bonus}.
In the opposite direction,
every {\nsr} proof after cutting the degrees and dropping
$r_{j\ell}$'s and $s_{jt}$'s provides a valid {\hitodd} proof.

Like {\ns} over $\GF(2)$, {\hitodd} can efficiently refute Tseitin formulas modulo 2 (see Def.~\ref{def:tseitin}),
which require exponential-size resolution proofs~\cite{Urq87}.
\begin{proposition}\label{prop:hitodd-tseitin}
For any constant-degree graph $G=(V,E)$ and 0/1-vector $c$,
{\hitodd} has a polynomial-size refutation of $T_{G,c}$.
\end{proposition}
\begin{proof}
Each truth assignment falsifies an odd number of vertex constraints.
For each constraint, it falsifies exactly one of the $2^{\deg v - 1}$ clauses.
Thus the total number of falsified clauses is odd,
and $T_{G,c}$ itself is an unsatisfiable odd-hitting formula.
\end{proof}

A separation between {\hitodd} and {\ns} without dual variables follows immediately from
the separation between {\nsr} and {\ns} of de~Rezende et al~\cite{dRLNS21}.

In the opposite direction, there are formulas
that require exponentially larger proofs in {\hitodd} than in {\res}.
Dmitry Sokolov [private communication] suggested that the well-known technique of xorification can
produce an exponential separation between the size of {\res} and {\nsr} proofs from the bounds of \cite{BOPCI00}:
\begin{theorem}[\cite{BOPCI00}]
  \label{th:BOPCI00}
  There exists a family of formulas that have {\res} proofs of constant width and require {\ns} degree $\Omega(n/\log n)$.
\end{theorem}

We notice that this technique is still viable for succinct {\nsr} proofs, and hence {\hitodd}.
In the following lemma we apply xorification and the random restriction technique of Alekhnovich and Razborov (see~\cite{Ben02}) to get the separation.

\begin{lemma}\label{lem:xorification}
  Let $F$ be a CNF formula that requires degree $d$ to refute in {\ns} over a field $\mathbb{F}$. Then $F \circ (\oplus_2)^n$ requires size $2^{\Omega(d)}$ to refute in succinct {\nsr} over $\mathbb{F}$.
\end{lemma}

\begin{proof}
  Denote by $y_{i,0}$ and $y_{i,1}$ the variables appearing in $F \circ (\oplus_2)^n$ as $y_{i,0}\oplus y_{i,1}$ instead of the variable $y_i$ in $F$.
  Let $\sum f_ig_i \equiv 1$ be a succinct {\nsr} refutation of $F \circ (\oplus_2)^n$, and let $s$ be the number of monomials in the refutation, which we assume for the sake of contradiction to be less than $(4/3)^d$.
  Let $\mathcal{D}$ be the probability distribution over random restrictions where for every pair of variables $y_{i,0},y_{i,1}$ we sample $j\in\set{0,1}$ and $b\in\set{0,1}$ uniformly and we assign $y_{i,j}=b$ while leaving $y_{i,1-j}$ unassigned.
  Observe that for $\rho\sim\mathcal{D}$ we have $\set{\restrict{f_i}{\rho}\mid f_i\in F\circ(\oplus_2)^n}=F$ and that $\rho$ respects Boolean axioms, therefore $\sum \restrict{f_ig_i}{\rho} = \sum \restrict{f_i}{\rho}\cdot\restrict{g_i}{\rho}$ is a {\nsr} refutation of $F$.

  Since for any monomial $m$ we have
  \[
    \Pr_{\rho\sim\mathcal{D}}[\deg(\restrict{m}{\rho}) \geq d] \leq
    \Pr_{\rho\sim\mathcal{D}}[\deg(m) \geq d \text{ and } \restrict{m}{\rho} \neq 0] \leq (3/4)^d,
  \]
  the union bound gives that
  \[
    \Pr_{\rho\sim\mathcal{D}}[\deg(\restrict{f_ig_i}{\rho}) \geq d] \leq
    s \Pr_{\rho\sim\mathcal{D}}[\deg(\restrict{m}{\rho}) \geq d] \leq
    s \cdot(3/4)^d < 1.
  \]
  Therefore there exists a restriction such that $\deg(\restrict{f_ig_i}{\rho}) < d$, contradicting our hypothesis that $F$ requires {\ns} degree $d$.
\end{proof}

Combining xorification with a lower bound on the degree of pebbling formulas we obtain a separation between {\hitodd} and {\res}.

\begin{corollary}
  \label{cor:pebbling-ns}
  There exists a family of formulas that have {\res} proofs of polynomial size and require {\hitodd} proofs of size $2^{\Omega(n/\log n)}$.
\end{corollary}
\begin{proof} By Theorem~\ref{th:BOPCI00} and Lemma~\ref{lem:xorification}.\end{proof}

\section{\texorpdfstring{\hitp}{Hitting(xor)}}
\label{sec:hitp-lowerbound}
% \RED{To justify the purpose of this section we should explain why an analogue of Theorem~\ref{th:tres-qsim-hit} does not (obviously) hold. Otherwise this separation could be obtained more easily from \cite{IR21} plus that.} \textcolor{magenta}{AR: see proposition 6.1} \textcolor{magenta}{MV: nice!}

{\hitp} extends {\hit} to formulas that work with linear equations modulo two.
We know from Cor.~\ref{cor:tlresxor-vs-hitting} that Tseitin formulas
separate {\hitp} from {\hit} and {\res}.

In this section, we show that perfect matching formulas (that have polynomial-size {\cp} proofs) require exponential-size
{\hitp} refutations. In order to do this, we lift them using (binary) xorification
and then reduce the question to the known communication complexity lower bound 
for set disjointness. 

\subsection{Evidence against quasi-polynomial simulation by \texorpdfstring{{\tresp}}{tl-Res(+)}}
\label{sec:no-qp-simulation}
Theorem~\ref{th:tres-qsim-hit} suggests that there might be a quasi-polynomial simulation of {\hitp} by {\tresp}, which would imply that all exponential-size {\tresp} lower bounds imply exponential-size lower bounds for {\hitp}. Recall that the crux of the proof of Theorem~\ref{th:tres-qsim-hit} is that we can always find a literal that appears in an $\Omega(1/\log n)$ fraction of the clauses of a hitting formula, and this provides an efficient splitting of the problem. The following proposition shows that the analogue of this does not exist for {\hitp}, which is a piece of evidence that there is no similar simulation in the case of {\hitp}. Namely, we prove that, contrary to the case of {\hit}, there are formulas that cannot be split so efficiently.

We use the construction that appears in \cite[Section 5.1]{avsp} for a different purpose.
\begin{proposition}
    There exists an unsatisfiable hitting$(\oplus)$ formula over $2t-1$ variables consisting of $2^t$ $(t-1)$-dimensional affine subspaces such that for every set $S \subseteq [2t-1]$ there is at most one affine space in the formula contained in the codimension-1 affine subspace $\left\{x \in\{0,1\}^{2t-1} \mid \bigoplus_{i \in S} x_i = \alpha\right\}$ for $\alpha \in \{0,1\}$.
\end{proposition}
\begin{proof}
   Let $\mathbb{F}_{2^t}$ be the finite field of order $2^t$, which we can identify with polynomials of degree smaller than $t$ over $\mathbb{F}_2$. Given such a polynomial $b_0 + b_1 x + \dotsb + b_{t-1} x^{t-1}$, we can identify it with the string $b_0 b_1 \ldots b_{t-1}$.

For every $\alpha \in \mathbb{F}_{2^t}$, we define
\[
 V_\alpha := \{ (c, \alpha c) \mid c = c_0 + c_1 x + \dotsb + c_{t-1} x^{t-1} \in \mathbb{F}_{2^t}, c_0 = 1 \}.
\]
We identify $V_\alpha$ with a subset of $\{0,1\}^{2t-1}$ by converting both parts to strings, concatenating them, and removing the initial~$1$.

\begin{claim}
Sets $\{V_\alpha\}_{\alpha \in \mathbb{F}_{2^t}}$ form an unsatisfiable hitting$(\oplus)$ formula.
\end{claim} 
\begin{proof}
Define
\[
 U_\alpha := \{ (d, \alpha d) \mid d \in \mathbb{F}_{2^t} \}.
\]
It is easy to see that $U_\alpha$ is a vector subspace of $\mathbb{F}_2^{2t}$ as $\eta_1(d_1, \alpha d_1) + \eta_2(d_2, \alpha d_2) = ((\eta_1 d_1 + \eta_2 d_2), \alpha (\eta_1 d_1 + \eta_2 d_2)) \in U_\alpha$. If $(x,y) \in U_\alpha \cap U_\beta$ for $\alpha \neq \beta$ then $y = \alpha x = \beta x$, and so $x = y = 0$. This shows that $U_\alpha,U_\beta$ are trivially intersecting subspaces. Also, if $x \neq 0$ then $(x,y)$ is covered by $U_{y/x}$. Together with $U_\infty = \{ (0,d) \mid d \in \mathbb{F}_{2^t} \}$, we obtain the object called \emph{standard Desarguesian spread} of $\mathbb{F}_{2^t}^2 \setminus \{(0,0)\}$ (which we think of as a projective plane).

We obtain $V_\alpha$ by restricting $U_\alpha$ to those $d$ satisfying $d_0 = 1$. Thus $V_\alpha \cap V_\beta = \emptyset$ and the $V_\alpha$'s cover all points of the form $(x,y)$ where $x_0 = 1$. Moreover, $V_\alpha$ is an affine subspace since it is an intersection of $U_\alpha$ and an affine space $\{(c,d) \mid c_0 = 1\}$ projected on the last $2t-1$ coordinates. We conclude that the $V_\alpha$'s constitute a hitting$(\oplus)$ formula.
\end{proof}

Now we proceed to prove the uniqueness of an affine space $V_{\alpha}$ contained in an arbitrary codimension-1 affine subspace of $\{0,1\}^{2t-1}$.
We can think of the entire domain as $V = \{(c,d) : c,d \in \mathbb{F}_{2^t}, c_0 = 1\}$. The subspaces are $V_\alpha = \{(c,d) : d = \alpha c\}$. Any affine form on the encoding of length $2t-1$ (obtained by removing the initial $1$) corresponds to a linear form on the ``homogenized'' encoding of length $2t$. Moreover, we can express each such form as $(\gamma c + \delta d)_0$, where $\gamma,\delta \in \mathbb{F}_{2^t}$, and we take the coefficient of $x^0$ in the result.

A subspace $V_\alpha$ satisfies the constraint $(\gamma c + \delta d)_0 = 0$ if $((\gamma + \alpha \delta) c)_0 = 0$ for all $c \in \mathbb{F}_{2^t}$ such that $c_0 = 1$. Let $\eta = \gamma + \alpha \delta$. Choosing $c = 1$, we see that $\eta_0 = 0$. If $d \in \mathbb{F}_{2^t}$ satisfies $d_0 = 0$ then $(\eta d)_0 = (\eta (d + 1))_0 + \eta_0 = 0$. Therefore $(\eta c)_0 = 0$ for all $c \in \mathbb{F}_{2^t}$. In particular, if $\eta = 0$, since otherwise we obtain a contradiction by choosing $c = \eta^{-1}$.

There are now two cases. If $\delta = 0$ then $\eta = 0$ implies $\gamma = 0$. If $\delta \neq 0$ then $\eta = 0$ implies $\alpha = -\gamma/\delta$. We conclude that an affine equation that involves only the first half holds for no subspace (unless it is the trivial $0 = 0$), and any other affine equation holds for precisely one subspace.
\end{proof}

\subsection{Communication simulation of \texorpdfstring{{\hitp}}{Hitting(xor)}}
\label{sec:linear-hitting-communication}

Consider the following complexity measure on relations introduced by Jain and Klauck \cite[Definition 22]{JK10}. Let $S \subseteq \mathcal{X} \times \mathcal{Y} \times \mathcal{O}$, let $\recttwo{\mathcal{X}}{\mathcal{Y}}$ (or just $\rect$) be the set of subrectangles of $\mathcal{X} \times \mathcal{Y}$ (that is, $\{A\times B\ |\ A\subseteq \mathcal{X},\ B\subseteq \mathcal{Y}\}$), let $\supp(S) = \{(x,y) \in \mathcal{X} \times \mathcal{Y} \mid \exists o \in \mathcal{O}\colon (x,y,o) \in S\}$. Then $\varepsilon$-partition bound of $S$ is denoted by $\prt_\varepsilon(S)$ and defined by the value of the following linear program:
\begin{align*}
  & ~\sum_{o \in \mathcal{O}} \sum_{R \in \rect} w_{o, R} \to \min \\
   \forall (x,y) \in \supp(S), &~ \sum_{o: (x,y,o) \in S} \sum_{(x,y) \in R \in \rect} w_{o, R} \ge 1-\varepsilon \\
   \forall o, R\in\mathcal{O} \times \rect, &~ w_{o, R} \ge 0 \\
   \forall (x,y) \in \mathcal{X} \times \mathcal{Y}, &~ \sum_{o \in \mathcal{O}} \sum_{(x,y) \in R \in \rect} w_{o, R} = 1. \\
\end{align*}
That is, we put weights $w_{o,R}$ on every answer $o$ and rectangle $R$ so that for every question $(x,y)$ the total weight is 1 for all rectangles where $(x,y)$ participates and for all answers, and for every question that does have an answer, the weight $1-\varepsilon$ is concentrated on the correct answer(s).
\begin{remark}\label{rem:clause-search}
Note that if we aggregate the ``answers'' $o\in\mathcal{O}$ into larger containers
$I(o')\in\mathcal{O'}$ so that $\mathcal{O} = \mathop{{\mbox{\Large$\sqcup$}}}\limits_{o'\in\mathcal{O'}} I(o')$
and $(x,y,o')\in\mathcal{O'}\Leftrightarrow \exists o\in I(o')\ (x,y,o)\in\mathcal{O}$,
we can only decrease the value of $\prt_\varepsilon$.
\end{remark}

This measure lower bounds communication complexity in the following way:
\begin{theorem}[\cite{JK10}]
    $\varepsilon$-error randomized communication complexity of a relation $S$ is at least $\log \prt_\varepsilon (S)$.
\end{theorem}

The classical communication complexity lower bound for set disjointness works for this measure as well. The set disjointness relation $\disj_n\colon \{0,1\}^n \times \{0,1\}^n \to \{0,1\}$ is defined as $\disj_n(x,y) = \bigwedge_{i \in [n]} \lnot (x_i \land y_i)$. 

\begin{theorem}[\cite{Raz92, JK10}] \label{th:prt-disj-lb}
$\log \prt_\varepsilon(\disj_n) = \Omega(n)$.
\end{theorem}

We prove that {\hitp} size lower bounds can be deduced from lower bounds on $\prt_\varepsilon$.
Recall that {\hitp} can be considered not just a proof system for formulas in CNFs,
but as a proof system for sets of disjunctions of affine equations (sets of affine subspaces),
which we will still call \emph{clauses}.
We will use it in our proof (though the final bound will be for a formula in CNF).

\begin{lemma}
\label{lem:hitting-xor-reduction}
    Let $F = \bigwedge_{i \in [m]} C_i$ be a set of affine subspaces over $n$ variables having a {\hitp} refutation of size $s$. Let $X \sqcup Y = [n]$ be an arbitrary partition of variables of $F$ and let the clause-search relation associated with $F$ be \[S := \{(x,y,j) \in \{0,1\}^X \times \{0,1\}^Y \times [m] \mid C_j(x,y) = 0\}.\] Then for any constant $\varepsilon\in(0,1)$, $\prt_{\varepsilon} (S) = O(s^2)$ (the constant in $O$ may depend on $\varepsilon$).
\end{lemma}
\begin{proof}
 
   Let $A_1, \dots, A_s$ be the affine subspaces forming the partition of $\{0,1\}^n$ corresponding to the {\hitp} refutation of $F$, i.e.\ for every $A_j$ there exists $C_{h(j)} \in F$ such that for every $x \in A_j$, $C_{h(j)}(x) = 0$. 

    Now let us define the values of $w_{i, R}$ for $i \in [m]$ and a rectangle $R \in \recttwo{\{0,1\}^X}{\{0,1\}^Y}$ corresponding to $\prt_\varepsilon(S)$. 
    
    For each $j \in [s]$ we define a part of the weights induced by $A_j$. Let $A_j = \{(x,y) \in \{0,1\}^X \times \{0,1\}^Y \mid M^X x + M^Y y = a\}$ where $a \in \{0,1\}^k$, $M^X \in \{0,1\}^{[k] \times X}$, $M^Y \in \{0,1\}^{[k] \times Y}$. Consider a uniformly distributed matrix $V \in \{0,1\}^{t \times k}$, then if $(x,y) \in A_j$,
       \[\Pr_V \left[ V M^X x + V M^Y y = V a \right] = 1,\]
    and if $(x,y) \not\in A_j$,
       \[\Pr_V \left[ V M^X x + V M^Y y = V a \right] = \frac{1}{2^t}.\]
    For matrix $V$ and $u,v \in \{0,1\}^t$, define the rectangle
$R^V_{u,v} := \{(x,y) \mid V M^X x = u, V M^Y y = v\}$. We can now define weights $w^{(j)}_{i, R}$ (recall that $h(j)$ is the index of the clause falsified in $A_j$):
       \[ w^{(j)}_{i, R} := 
       \begin{cases} 
        2^{-kt} & \text{if } i=h(j) \text{ and } R = R^V_{u,u + V a}  \text{ for some } u \in \{0,1\}^t\text{ and } V \in \{0,1\}^{k \times t}, \\ 
        0 & \text{otherwise}. 
       \end{cases} \]
(Note that the points of $R_{u,u+Va}$ are guaranteed to fulfil the condition $V M^X x + V M^Y y = V a$ above.)
    Then
      \[ \sum_{\substack{i \in [m]\\ R \ni (x,y)}} w^{(j)}_{i, R} = \sum_{R \ni (x,y)} w^{(j)}_{h(j), R} =  \E_V w^{(j)}_{\displaystyle h(j),R^V_{\scriptstyle VM^{\!X} x, VM^{\!Y} y}} = \begin{cases} 1 & \text{if } (x,y) \in A_j, \\ 2^{-t} & \text{otherwise}. \end{cases}\]
    
    Now fix $t := \lceil\log_2 (2s/\varepsilon) \rceil$, so $2^{-t} \le \varepsilon/2s$, and set
     $w_{i, R} := \frac{1}{1 + (s-1) 2^{-t}} \sum_{j \in [s]} w^{(j)}_{i, R}$. Then
     \[ \sum_{\substack{i \in [m]\\ R \ni (x,y)}} w_{i,R} = \frac{1}{1 +  (s-1) 2^{-t}} \left(\overbrace{\sum_{\substack{i \in [m]\\ R \ni (x,y)}} w^{(j(x,y))}_{i, R}}^{1} + \sum_{j \in [s] \setminus \{j(x,y)\}} \overbrace{\sum_{\substack{i \in [m]\\ R \ni (x,y)}}  w^{(j)}_{i,R}}^{2^{-t}} \right) = 1,\]
   where $j(x,y) \in [s]$ is such that $A_{j(x,y)} \ni (x,y)$. Let $i(x,y)$ be the index of some fixed clause falsified in $A_{j(x,y)}$. Let us check the second condition on the weights:
   \[ \sum_{\substack{(x, y, i) \in S\\ R \ni (x,y)}} w_{i,R} \ge \!\sum_{R \ni (x,y)} w_{i(x,y), R} \ge \frac{1}{1 +  (s-1) 2^{-t}}  \!\sum_{R \ni (x,y)} w^{(j(x,y))}_{i(x,y), R} = \frac{1}{1 +  (s-1) 2^{-t}} \ge \frac{1}{1+\varepsilon/2} \ge 1-\varepsilon.\]
   The objective function of the linear program is
   \[\prt_\varepsilon (S) = \sum_{\substack{i \in [m]\\ R \in \rect}} w_{i,R} \le \sum_{\substack{i \in [m],\ j \in [s]\\R \in \rect}} w^{(j)}_{i,R} \ \le\! \ s 2^t \le 4 s^2/\varepsilon.\]
   The second inequality holds as by definition of $w^{(j)}_{i,R}$, for every $j\in[s]$, it is non-negative in the $2^{t+kt}$ cases, and equals $2^{-kt}$ in each of them.
\end{proof}

%-------------------------------------------------
\subsection{Lower bounds on \texorpdfstring{$\prt_\varepsilon$}{prt-eps}}\label{sec:linear-hitting-lb}

To get a lower bound on the size of {\hitp} refutations using Lemma~\ref{lem:hitting-xor-reduction}, we need to show a lower bound on $\prt_\varepsilon$ for clause-search relations. The idea is to utilize known reductions from set disjointness, $\disj_n$, that have been developed for randomized communication complexity.

The perfect matching principle $\PM_G$ is a formula in CNF stating that a given subset of $E(G)$ is a perfect matching. It has a variable $z_e$ for every edge $e\in E(G)$,
and for every $v\in V(G)$ it contains 
\begin{itemize}
\item the clauses $\bar{z}_f\lor\bar{z}_g$ for every pair of edges $f,g$ adjacent to $v$, and
\item the clause $\bigvee_{e\ni v}{z_e}$.
\end{itemize}

Define a set of clauses (with affine equations) 
$\PM^\oplus_G$ by replacing every variable $z_e$ in the formula $\PM_G$
by $x_e\oplus y_e$ (here $x_e$ and $y_e$ are variables of $\PM^\oplus_G$).

For a graph $G$, the relation 
 $\xorPM_G \subseteq \{0,1\}^{E(G)} \times \{0,1\}^{E(G)} \times V(G)$ 
for $\PM^\oplus_G$ is defined as
 $\xorPM_G := \{(E_1, E_2, v) \mid \text{degree of } v \text{ in } (V(G), E_1 \oplus E_2) \text{ is not } 1\}$. In other words, the $\xorPM_G$ problem asks to find a witness that the symmetric difference of the input sets is not a perfect matching. Strictly speaking, it does not ask for a specific clause; however, Remark~\ref{rem:clause-search} explains that aggregating clauses into the set of clauses related to a single vertex is enough.

In the following theorem, we construct a communication-complexity reduction of $\disj_n$ to $\xorPM_G$ for a specific graph $G$ by providing two functions $f_{\Alice}$, $f_{\Bob}$ for Alice and Bob and another function $g$ for recovering the result after $\xorPM_G$ finds a failing vertex in the graph $(V,f_{\Alice}(\ldots)\oplus f_{\Bob}(\ldots))$. We later use this reduction in order to prove bounds on $\prt$.

\begin{theorem}[{variation on \protect\cite[Theorem 16]{IR21}}]
\label{th:perfect-matching-reduction}
Let $G = K_{40n+1,40n+3}$. There exist a finite set $R$, functions $f_{\Alice}, f_{\Bob}\colon \{0,1\}^n \times R \to \{0,1\}^{E(G)}$ and a function $g\colon V(G) \times R \to \{0,1\}$ such that for every $a, b \in \{0,1\}^n$ and for every family of random variables $\{\mathbf{p}_{x,y}\}_{x,y \in \{0,1\}^{E(G)}}$ over $V(G)$,
\[ \Pr_{r \gets R} \left[g(\mathbf{p}_{f_{\Alice}(a,r), f_{\Bob}(b,r)},r) \neq \disj_n(a,b) \mid (f_{\Alice}(a, r), f_{\Bob}(b, r), \mathbf{p}_{f_{\Alice}(a,r), f_{\Bob}(b,r)}) \in \xorPM_G \right] \le \frac{1}{10},\]
where $r$ is uniformly distributed over $R$.
\end{theorem}
\begin{figure}
\begin{center}
    \input{xortable_explicit.tex}
    \caption{The graphs $A(0), A(1), B(0)$, and $B(1)$ and their pairwise symmetric differences. Only $A(1) \oplus B(1)$ is not a matching.}
    \label{fig:xortable}
\end{center}
\end{figure}
\begin{proof}
Itsykson and Riazanov~\cite{IR21} described four graphs $A(0),A(1),B(0),B(1)$, all of them subgraphs of $K_{4,4}$, with the following two properties (see Fig.~\ref{fig:xortable}): 
\begin{itemize}
\item $A(a) \oplus B(b)$ is a perfect matching consisting of four edges for all $(a,b) \neq (1,1)$.
\item $A(1) \oplus B(1)$ is the disjoint sum of $K_{1,3}$ and $K_{3,1}$, where in both cases the vertex numbered $1$ on one side is connected to the vertices numbered $2,3,4$ on the other side.
\end{itemize}

Using this gadget, we define $R,f_{\Alice},f_{\Bob},g$:
\begin{itemize}
\item $R = S_{40n+1} \times S_{40n+3} \times \{0,1\}^{E(G)}$, where $S_m$ is the symmetric group on $m$ elements. In what follows, we denote an element of $R$ by $(\pi_1,\pi_2,H)$.

\item $f_{\Alice}(a,\pi_1,\pi_2,H)$: start with a graph $G_A$ that is a disjoint union of $10$ copies of $A(a_1),\dotsc,A(a_n)$ (total of $10n$ subgraphs) and a copy of $K_{1,3}$ with its vertex on the left numbered $40n+1$. 
Compute the symmetric difference with $H$, then
permute the vertices on the left using $\pi_1$ and the vertices on the right using $\pi_2$.

\item $f_{\Bob}(b,\pi_1,\pi_2,H)$: start with a graph $G_B$ that is a disjoint union of $10$ copies of  $B(b_1),\dotsc,B(b_n)$ and a copy of the complement of $K_{1,3}$ (that is, the empty graph $1\times 3$).
As for $f_{\Alice}$, we compute the symmetric difference with $H$ and then apply $\pi_1,\pi_2$.

\item $g(v,\pi_1,\pi_2,H)$: return $1$ if $v$ is vertex $\pi_1(40n+1)$ on the left. %Since $g$ only depends on $v$ and $\pi_1$, we think of it as a function of these two arguments alone.

\end{itemize}

If $\disj_n(a,b) = 1$ then the graph $f_{\Alice}(a,r) \oplus f_{\Bob}(b,r)$ consists of a matching of $40n$ edges together with a copy of $K_{1,3}$, after shuffling both sides according to $r$. The only vertex of degree other than~$1$ is $\pi_1(40n+1)$, and so if $\mathbf{p}$ returns a vertex in $\xorPM$, it must return $\pi_1(40n+1)$. Consequently, $g$ always returns $1$. In this case, the probability in the statement of the theorem is zero.

Now suppose that $\disj_n(a,b) = 0$, say $a_1 = b_1 = 1$. Let $G_{\oplus} = G_A \oplus G_B$, where $G_A,G_B$ are the graphs in the definitions of $f_{\Alice},f_{\Bob}$, and define 
the permutation that touches only vertices in the ten copies of $A(1)\oplus B(1)$ and the vertices of the final $K_{1,3}$.
\begin{align*}
\sigma_1 &= (1 \; 5 \; 9 \; \cdots \; 37 \; 40n+1), \\
\sigma_2 &= (2 \; 6 \; 10 \; \cdots \; 38 \; 40n+1) \, (3 \; 7 \; 11 \; \cdots \; 39 \; 40n+2) \,
 (4 \; 8 \; 12 \; \cdots \; 40 \; 40n+3).
\end{align*}
Thus $G_\oplus$ is invariant under permuting vertices on the left using $\sigma_1$ and those on the right using $\sigma_2$ (simultaneously).

For fixed $\pi_1,\pi_2$, the distribution of $(f_{\Alice}(a,r),f_{\Bob}(b,r))$ (whose randomness comes only from $H$) can be described as follows: sample a pair $(H_A,H_B)$ whose symmetric difference is $G_\oplus$, and then permute the result according to $\pi_1,\pi_2$, an operation we denote by superscripting the pair. Therefore
\begin{gather*}
\Pr_{r \gets R} \left[g(\mathbf{p}_{f_{\Alice}(a,r), f_{\Bob}(b,r)},r) \neq \disj_n(a,b) \mid (f_{\Alice}(a, r), f_{\Bob}(b, r), \mathbf{p}_{f_{\Alice}(a,r), f_{\Bob}(b,r)}) \in \xorPM_G \right] = \\
\Pr_{\substack{\pi_1,\pi_2 \\ H_A \oplus H_B = G_\oplus}}
\left[ \mathbf{p}_{H_A^{\pi_1,\pi_2},H_B^{\pi_1,\pi_2}} = \pi_1(40n+1) \mid (H_A^{\pi_1,\pi_2},H_B^{\pi_1,\pi_2},\mathbf{p}_{H_A^{\pi_1,\pi_2},H_B^{\pi_1,\pi_2}}) \in \xorPM_G \right].
\end{gather*}
Denote $S(\pi_1, \pi_2) := \bigl(H_A^{\pi_1,\pi_2},H_B^{\pi_1,\pi_2},\mathbf{p}_{H_A^{\pi_1,\pi_2},H_B^{\pi_1,\pi_2}}\bigr)$.
Let $\mathbf{t}$ be chosen uniformly at random from $\{0,\dotsc,10\}$. 
Since $(\pi_1,\pi_2)$ and $(\pi_1\sigma_1^{\mathbf{t}},\pi_2\sigma_2^{\mathbf{t}})$ have the same distribution, the probability above is equal to
\[
\begin{aligned}
 \Pr_{\substack{\pi_1,\pi_2 \\ H_A \oplus H_B = G_\oplus}}
&\left[ \mathbf{p}_{H_A^{\pi_1 \sigma_1^{\mathbf{t}},\pi_2 \sigma_2^{\mathbf{t}}},H_B^{\pi_1 \sigma_1^{\mathbf{t}},\pi_2 \sigma_2^{\mathbf{t}}}} = \pi_1 \bigl(\sigma_1^{\mathbf{t}}(40n+1))\ \middle|\ S(\pi_1 \sigma_1^{\mathbf{t}},\pi_2 \sigma_2^{\mathbf{t}}) \in \xorPM_G \right] = \\
 \Pr_{\substack{\pi_1,\pi_2 \\ H_A \oplus H_B = G_\oplus^{\sigma_1^{\mathbf{t}},\sigma_2^{\mathbf{t}}}}}
&\left[ \mathbf{p}_{H_A^{\pi_1 ,\pi_2 },H_B^{\pi_1,\pi_2}} = \pi_1(\sigma_1^{\mathbf{t}}(40n+1))\ \middle|\ S(\pi_1,\pi_2) \in \xorPM_G \right] = \\
\Pr_{\substack{\pi_1,\pi_2 \\ H_A \oplus H_B = G_\oplus}}
&\left[ \mathbf{p}_{H_A^{\pi_1 ,\pi_2 },H_B^{\pi_1,\pi_2}} = \pi_1(\sigma_1^{\mathbf{t}}(40n+1))\ \middle|\ S(\pi_1,\pi_2) \in \xorPM_G \right] \leq \frac{1}{11},
\end{aligned}
\]
since $\sigma_1^{\mathbf{t}}(40n+1)$ is uniformly distributed over $\{1, 5, 9, \dotsc, 37, 40n+1\}$.
\end{proof}

In the next theorem we use the reduction from Theorem~\ref{th:perfect-matching-reduction}
in order to transform $\prt$ bounds for $\xorPM_G$ into bounds for the disjointness.
\begin{theorem} \label{th:PM-to-disj}
For the graph $G$ from Theorem~\ref{th:perfect-matching-reduction} and small enough constant $\varepsilon$, \[\prt_{1/3}(\disj_n) \le \prt_\varepsilon(\xorPM_G). \]
\end{theorem}
\begin{proof}
Let $w_{v,R}$ for $v \in V(G)$ and $R\in \rectone{\{0,1\}^{E(G)}}$ be optimal weights in the $\varepsilon$-partition bound for $\xorPM_G$. Let $f_{\Alice}, f_{\Bob}, g$ be the reduction functions from Theorem~\ref{th:perfect-matching-reduction}. Let $f_{\Alice}^r(x) := f_{\Alice}(x, r)$ and $f_{\Bob}^r(x) := f_{\Bob}(x,r)$. We claim that the following weights $w'_{\alpha, R}$ for $\alpha \in \{0,1\}$ and $R \in \rectone{\{0,1\}^n}$ yield the upper bound on $\prt_{1/3}(\disj_n)$:
\[ w'_{\alpha, X \times Y} := \E_r  \sum_{\substack{A, B\colon\\ \bigcup (f_{\Alice}^r)^{-1} (A) = X\\ \bigcup (f_{\Bob}^r)^{-1} (B) = Y}} \sum_{\substack{v\colon\\ g(v,r)=\alpha}} w_{v,A\times B}.\]

Let us verify the properties of $w'$. First, we check that the sum of weights of rectangles covering a point is exactly 1:
\[ \sum_{\substack{\alpha \in \{0,1\}\\ X \times Y \ni (x,y)}} w'_{\alpha, X \times Y} = \E_r\sum_{v \in V(G)} \sum_{\substack{A,B\colon \\A \ni f_{\Alice}^r(x)\\[+1pt]B\ni f_{\Bob}^r(y)}} w_{v,A \times B} = 1.\]
Now we check that the sum of weights covering a point $(x,y)$ and labeled with the correct answer $\alpha = \disj_n(x,y)$ is at least $2/3$:
\begin{multline*}
    \sum_{\substack{ X \times Y \ni (x, y)}} w'_{\alpha, X \times Y} = \E_s \sum_{\substack{A \ni f_{\Alice}^r(x)\\[+1pt]B\ni f_{\Bob}^r(y) \\ v\colon g(v, r) = \alpha}} w_{v, A \times B}   
    \ge \E_s \sum_{\substack{A \ni f_{\Alice}^r(x),\ B\ni f_{\Bob}^r(y) \\ v\colon g(v, r) = \alpha \\ (f_{\Alice}^r(x), f_{\Bob}^r(y), v) \in \xorPM_G}} w_{v, A \times B}   \\
    = \E_s \sum_{\substack{A \ni f_{\Alice}^r(x),\ B\ni f_{\Bob}^r(y) \\ (f_{\Alice}^r(x), f_{\Bob}^r(y), v) \in \xorPM_G}} w_{v, A \times B}   - \E_s \sum_{\substack{A \ni f_{\Alice}^r(x),\ B\ni f_{\Bob}^r(y) \\ v\colon g(v, r) \neq \alpha \\ (f_{\Alice}^r(x), f_{\Bob}^r(y), v) \in \xorPM_G}} w_{v, A \times B}   
%   \\ \ge 1 - \varepsilon - \E_s \sum_{\substack{A \ni f_{\Alice}^r(x),\ B\ni f_{\Bob}^r(y) \\ v\colon g(v, r) \neq \alpha \\ (f_{\Alice}^r(x), f_{\Bob}^r(y), v) \in \xorPM_G}} w_{v, A \times B}   \\
    %=
\\\ge
1 - \varepsilon - \underbrace{\E_s \sum_{\substack{v\colon g(v, r) \neq \alpha \\ (f_{\Alice}^r(x), f_{\Bob}^r(y),v) \in \xorPM_G}} \sum_{\substack{A \ni f_{\Alice}^r(x)\\[+1pt]B\ni f_{\Bob}^r(y)}} w_{v, A \times B}}_{r}.
\end{multline*}
Now let us define a family of random variables $\mathbf{p}_{a,b}$ for $a,b\in \{0,1\}^{E(G)}$ over $V(G)$ such that $\Pr[\mathbf{p}_{a,b} = v] = \sum_{A \times B \ni (a,b)} w_{v,A\times B}$ and rewrite
\begin{multline*}    
\E_r \sum_{\substack{v\colon g(v, r) \neq \alpha \\ (f_{\Alice}^r(x), f_{\Bob}^r(y),v) \in \xorPM_G}} \Pr[\mathbf{p}_{f_{\Alice}^r(x), f_{\Bob}^r(y)}=v] \\=\Pr_r\left[g(\mathbf{p}_{f_{\Alice}^r(x), f_{\Bob}^r(y)}, r) \neq \alpha \land (x,y,\mathbf{p}_{f_{\Alice}^r(x), f_{\Bob}^r(y)}) \in \xorPM_G\right] 
    \le\frac{1}{10}.
\end{multline*}
The last inequality holds by Theorem~\ref{th:perfect-matching-reduction}, so for $\varepsilon < \frac{9}{10} - \frac{2}{3}$ the condition $\sum\limits_{\substack{ X \times Y \ni (x, y)}} w'_{\alpha, X \times Y} \ge \frac23 $ holds. Now let us compute the sum of all $w'$:
\[ \sum_{\alpha \in \{0,1\}; X, Y \subseteq \{0,1\}^{n}} w'_{\alpha, X \times Y} = \E_s  \sum_{X, Y \subseteq \{0,1\}^n; v \in V(G)} w_{v, f_{\Alice}(A,r) \times f_{\Bob}(B,r)}   \le \prt_\varepsilon (\xorPM_G). \qedhere\]

\end{proof}

%-------------------------------------
\subsection{A lower bound on the size of \texorpdfstring{{\hitp}}{Hitting(xor)} refutations}
We can now derive the lower bound. 

\begin{theorem}\label{th:hitting-plus-lb-pm}
Any {\hitp} refutation of $\PM_G$ for the graph $G$ from Theorem~\ref{th:perfect-matching-reduction} contains $2^{\Omega(n)}$ many subspaces.
\end{theorem}
\begin{proof}
Let $\varepsilon$ be the constant from Theorem~\ref{th:PM-to-disj}. If there is a {\hitp}
refutation of $\PM_G(\bar{z})$ with $s$ subspaces, then we can convert it to 
a {\hitp} refutation of $\forxorPM_G(\bar{x},\bar{y})$ (viewed
as a set of subspaces) with the same number of subspaces
simply by replacing every variable $z_e$ with $x_e\oplus y_e$
in the hitting$(\oplus)$ formula.
Indeed, if a subspace $M\bar{z}=\bar{a}$ implies that 
a clause related to a vertex $v$ is false, 
then the subspace $M(\bar{x}+\bar{y})=\bar{a}$ 
implies that the degree of $v$ in $(V(G),E_1 \oplus E_2)$ is not~$1$
(that is, one of the clauses of $\forxorPM_G$ related to $v$ is false).
Lemma~\ref{lem:hitting-xor-reduction} (note also Remark~\ref{rem:clause-search})
then implies that 
$\prt_\varepsilon(\xorPM_G) = O(s^2)$, and so 
$\prt_{1/3}(\disj_n) = O(s^2)$ according to Theorem~\ref{th:PM-to-disj}. 
Theorem~\ref{th:prt-disj-lb} implies that $\log s = \Omega(n)$.
\end{proof}

\begin{remark}
Note that the $\PM_G$ formulas for $K_{i,j}$ for $i\neq j$
have polynomial-size {\cp} proofs: it can be easily derived from the 2-clauses
that the number of edges around a vertex is at most 1; then take the sum
of such inequalities around all vertices in the smaller part, and
take the sum of the other input inequalities in the larger part.
\end{remark}

%=============================================================================

%=============================================================================
\section*{Acknowledgements}
We thank Jan Johannsen, Ilario Bonacina, Oliver Kullmann, and Stefan Szeider for introducing us to the topic; Zachary Chase, Susanna de Rezende, Mika Göös, Amir Shpilka, and Dmitry Sokolov for helpful discussions. We also thank anonymous reviewers for their comments that helped us to improve the presentation.

This project has received funding from the European Union's Horizon 2020 research and innovation programme under grant agreement No~802020-ERC-HARMONIC.

\bibliographystyle{alphaurl}
\bibliography{hitting-proof}

\end{document}

%% file: xortable_explicit.tex
\newcommand{\printlayout}{%
\node (L1) {};
\node[right of=L1] (L2) {};
\node[right of=L2] (L3) {};
\node[right of=L3] (L4) {};
\node[below of=L1] (R1) {};
\node[right of=R1] (R2) {};
\node[right of=R2] (R3) {};
\node[right of=R3] (R4) {};}
\begin{tabular}{rccc}
     & & $B(0)$ & $B(1)$ \\
     & - & \boxed{\scalebox{0.3}{\begin{tikzpicture}[node distance=20mm, auto, every node/.append style={circle, draw=none, fill=black}]
    \printlayout
    \draw (L1) -- (R3);
    \draw (L2) -- (R2);
    \draw (L3) -- (R1);
    \draw (L4) -- (R4);
    \end{tikzpicture}}} & 
    \boxed{\scalebox{0.3}{\begin{tikzpicture}[node distance=20mm, auto, every node/.append style={circle, draw=none, fill=black}]
    \printlayout
    \draw (L1) -- (R2);
    \draw (L1) -- (R3);
    \draw (L1) -- (R4);
    \draw (R1) -- (L2);
    \draw (R1) -- (L3);
    \draw (R1) -- (L4);    
    \end{tikzpicture}}}\\
    $A(0)$ & \boxed{\scalebox{0.3}{\begin{tikzpicture}[node distance=20mm, auto, every node/.append style={circle, draw=none, fill=black}]
    \printlayout
    \draw (L1) -- (R2);
    \draw (L1) -- (R3);
    \draw (L2) -- (R1);
    \draw (L2) -- (R2);
    \draw (L3) -- (R3);
    \draw (L3) -- (R1);
    \end{tikzpicture}}} & 
    \boxed{\scalebox{0.3}{\begin{tikzpicture}[node distance=20mm, auto, every node/.append style={circle, draw=none, fill=black}]
    \printlayout
    \draw (L1) -- (R2);
    \draw (L2) -- (R1);
    \draw (L3) -- (R3);
    \draw (L4) -- (R4);
    \end{tikzpicture}}}& 
    \boxed{\scalebox{0.3}{\begin{tikzpicture}[node distance=20mm, auto, every node/.append style={circle, draw=none, fill=black}]
    \printlayout
    \draw (L1) -- (R4);
    \draw (L2) -- (R2);
    \draw (L3) -- (R3);
    \draw (L4) -- (R1);
    \end{tikzpicture}}}\\
    $A(1)$ & \boxed{\scalebox{0.3}{\begin{tikzpicture}[node distance=20mm, auto, every node/.append style={circle, draw=none, fill=black}]
    \printlayout
    \end{tikzpicture}}} & \boxed{\scalebox{0.3}{\begin{tikzpicture}[node distance=20mm, auto, every node/.append style={circle, draw=none, fill=black}]
    \printlayout
    \draw (L1) -- (R3);
    \draw (L2) -- (R2);
    \draw (L3) -- (R1);
    \draw (L4) -- (R4);
    \end{tikzpicture}}} & 
    \boxed{\scalebox{0.3}{\begin{tikzpicture}[node distance=20mm, auto, every node/.append style={circle, draw=none, fill=black}]
    \printlayout
    \draw (L1) -- (R2);
    \draw (L1) -- (R3);
    \draw (L1) -- (R4);
    \draw (R1) -- (L2);
    \draw (R1) -- (L3);
    \draw (R1) -- (L4);    
    \end{tikzpicture}}} \\
\end{tabular}